\documentclass[11pt,a4paper,pdftex]{scrartcl}
\usepackage[utf8]{inputenc}  
\usepackage[T1]{fontenc}
\usepackage[english]{babel}
\usepackage{amsmath}
\usepackage{amsthm}
\usepackage{amsfonts}
\usepackage{amssymb}
\usepackage[pdftex]{graphicx}
\usepackage{booktabs}
\usepackage{url}
\usepackage{color}
\usepackage{enumerate}
\usepackage{dsfont}
\usepackage{tikz}

\newcommand{\R}{\mathbb{R}}
\newcommand{\N}{\mathbb{N}}

\newcommand{\diff}{\operatorname{d}\!}

\newcommand{\C}{\mathcal{C}}
\newcommand{\D}{\mathcal{D}}

\newcommand{\scalprod}[2]{\left\langle #1, #2 \right\rangle}
\newcommand{\norm}[1]{\left\| #1 \right\|}

\newcommand{\bN}{\mathbf{N}}

\DeclareMathOperator{\const}{Const}
\newcommand{\mcal}{\mathcal}

\newcommand{\trec}{t_{\operatorname{rec}}}

\usepackage[autostyle]{csquotes}

\usepackage[hidelinks]{hyperref}
\ifdefined\hypersetup
\hypersetup{
	pdfkeywords={}, linkcolor=black, citecolor=black, filecolor=black, urlcolor=black,
}
\fi
\usepackage{aliascnt}
\usepackage[capitalise]{cleveref}
\let\oldtheorem\newtheorem
\RenewDocumentCommand{\newtheorem}{s m o m O{}}{%
\IfBooleanTF{#1}%
{\oldtheorem{#2}{#4}}%
{\IfNoValueTF{#3}{\oldtheorem{#2}{#4}[#5]}%
{\newaliascnt{#2}{#3}%
\oldtheorem{#2}[#2]{#4}%
\aliascntresetthe{#2}}}}

\theoremstyle{plain}
\newtheorem{theorem}{Theorem}[section]
\newtheorem{proposition}[theorem]{Proposition}
\newtheorem{corollary}[theorem]{Corollary}
\newtheorem{lemma}[theorem]{Lemma}

\theoremstyle{definition}
\newtheorem{definition}[theorem]{Definition}

\theoremstyle{remark}

\newtheorem{remark}[theorem]{Remark}

\crefformat{theorem}{Thm.~#2#1#3}
\crefformat{lemma}{Lem.~#2#1#3}
\crefformat{definition}{Defn.~#2#1#3}
\crefformat{equation}{(#2#1#3)}
\crefformat{appendix}{App.~#2#1#3}
\crefformat{question}{Question~#2#1#3}
\crefformat{proposition}{Prop.~#2#1#3}

\begin{document}

\title{Blow-Up, Asymptotics, and Improbability of Collision Singularities in the n-Body Problem}
\author{Nathan Duignan\thanks{School of Mathematics and Statistics, University of Sydney, Sydney, Australia} \and Manuel Quaschner\thanks{Institut für Mathematik, Friedrich Schiller Universität Jena, 07737 Jena, Germany}}
\maketitle

\begin{abstract}
We study finite-time singularities in $n$-body systems with pair potentials that are homogeneous of degree $-\alpha$, where $0<\alpha<2$, that are not necessarily attractive and with bounded time-dependent external forces. 
We develop a systematic method to analyse a colliding subsystem in the presence of simultaneous collision or non-collision singularities elsewhere in the system.
We extend classical results of von Zeipel and Painlev\'e to such subsystems, establish bounds on their internal energy near total collision, and implement a McGehee blow-up that remains effective without conservation of energy. The blow-up yields the asymptotics of the cluster size, and in the case of attractive potentials, the convergence towards the set of central configurations and asymptotics of relative position, relative velocities, and internal angular momentum. 
As an application of the asymptotic results, we construct a simplified proof that the initial conditions leading to collision form a set of measure zero.
\end{abstract}

\section{Introduction}
\label{sec:Introduction}

The $n$-body problem is one of the classical problems in mathematical physics. The problem can be stated simply; given $n$ masses with initial conditions for the positions and velocity of each body, determine the consequent motion of the bodies under a given potential $V$. The most classical potential is the gravitational Newtonian potential. Despite centuries of dedicated research, there remains many open problems concerning the motion of the $n$ bodies. In this paper, we focus on the possible \emph{singular orbits} of the $n$-body problem, i.e. when solutions only exist for a finite time. 

We investigate the dynamics of $n$ bodies under pair potentials on $\R^d$ of the form
\[
    V(q) = \sum_{j<k} V_{jk}(q),\qquad  V_{jk}(q)=\frac{Z_{jk}}{\norm{q_j-q_k}^{\alpha}},
    \qquad j<k,\ Z_{jk}\in\R,\ 0<\alpha<2,
\]
where $q_j \in \R^d$ is the position of the $j$-th body. In addition, we allow for bounded time-dependent external forces $F_j(t)$ on each body. Consequently, if $p_j\in \R^d$ is the momentum of the $j$-th body and $m_j$ the mass, the equations of motion are 
\begin{equation}\label{eq:Motion}
	\dot{q}_j = \frac{p_j}{m_j},\qquad \dot{p}_j = -\frac{\partial V}{\partial q_j} + F_j(t),\qquad j = 1,\dots, n 
\end{equation}
Newtonian gravity is the case $\alpha=1$, $Z_{jk}=-m_jm_k$, with no external force. However, the general assumption that $Z_{jk}\in\R$ permits simultaneous proofs of results for more general systems such as the dynamics of point charges.

Finite-time singularities are a central obstruction to understanding the global dynamics of the $n$-body problem. Collision singularities occur when two or more bodies approach one another, while for $n>3$ there may also be non-collision singularities. Painlev\'e conjectured the existence of non-collision singularities in the Newtonian problem in 1897 \cite{painleve1897leccons}, Xia established their existence for $n \geq 5$ almost a century later \cite{xia1992existence}, and the four-body case was resolved more recently \cite{xue2020non}. The classical theory of finite-time singularities begins with Painlev\'e's theorem, which shows that the minimum mutual distance tends to zero at any finite-time singularity \cite{painleve1897leccons}, and von Zeipel's theorem, which relates bounded singular motion to collision \cite{von1908singularites,sperling1970real} (see also the exposition of McGehee \cite{mcgeheeZeipelsTheoremSingularities1986}).

The asymptotic study of collision singularities originates with Sundman and Chazy and was developed systematically by Wintner and Sperling \cite{sundmanMemoireProblemeTrois1913,chazyCertainesTrajectoiresProbleme1918,wintner1941analytical,sperling1970real,sperling1972collision}. For total collisions of all $n$ bodies under a Newtonian potential, this theory gives the sharp moment-of-inertia law
$ I(t)\sim A(T^+-t)^{4/3}$,
and consequently the corresponding rates for kinetic and potential energy, and that the normalised configuration approaches the set of central configurations. Sperling extended the asymptotic theory to partial and simultaneous collisions, where the subsystem energy and momenta are no longer conserved \cite{sperling1970real,sperling1972collision} (see also the work in \cite{gierzkiewiczNoInfiniteSpin2025}). McGehee provided a geometric approach to the study of collision asymptotics by blowing up the singularity to an invariant collision manifold and desingularising time \cite{mcgehee1974triple}. This geometric approach has since been developed in many settings \cite{saari1981manifolds,moeckel1983orbits,elbialy1990collision,duignanChazyTypeAsymptoticsHyperbolic2020,duignan2021c,moeckel2023total}. ElBialy in particular used such methods to obtain refined asymptotics for simultaneous Newtonian collision clusters \cite{elbialy1990collision}, including that each cluster approaches a central configuration. This extension is non-trivial; care must be taken as, for example, even simultaneous binary collisions need not possess the smooth regularisability enjoyed by an isolated binary collision \cite{simoRegularizationSimultaneousBinary1992,martinezDegreeDifferentiabilityRegularization2000,duignanC83regularisationSimultaneous2020,duignan2021c}.

Collision analysis has also been extended beyond the Newtonian potential. Early examples include Devaney's work on anisotropic Kepler and more general mechanical singularities \cite{devaneyCollisionOrbitsAnisotropic1978,devaneyBlowingSingularitiesClassical1982} and McGehee's study of inverse-power binary interactions \cite{mcgeheeDoubleCollisionsClassical1981}. The asymptotic results have also been developed for homogeneous, quasi-homogeneous and time-dependent singular potentials \cite{barutelloSingularitiesGeneralizedSolutions2008}, with related work on quasi-homogeneous total collisions in \cite{diacuCentralConfigurationsTotal2006}. These general theories, however, retain an attractive or `collision-coercive' leading potential, so that the singular potential has a definite sign and diverges on approach to collision. Non-attractive collision dynamics has been studied in several charged few-body models \cite{atelaChargedIsosceles3body1988,sanoClassicalCoulombThreeBody2005,MR1816897}, but these analyses rely on low dimension, symmetry, or total collision. Recently in \cite{fleischer2018improbability,fleischer2019improbability}, the classical theorems of Painlev\'e and von Zeipel have been extended to some non-attractive potentials, however, the results hold for the entire system (not subsystems) and without bounded external forces (which allows for conservation of energy thus some simplifications of proofs).

Thus there is no comparable general theory for either the qualitative dynamics or asymptotics of finite-time singular orbits that allows for arbitrary simultaneous collision clusters, possible simultaneous non-collision singularities, and truly non-attractive potentials. In this work, we develop a general framework for the analysis of finite-time singularities that encompasses $-\alpha$-homogeneous potentials with arbitrary coefficients and with bounded external force.

Our starting point is to observe that a singular orbit admits a \emph{final cluster decomposition} \cite{quaschner2025improbability}. The final cluster decomposition groups together bodies who come close to one another (maybe indirectly) as $t$ approach the finite-time singularity $T^+$ (see \cref{def:FinalClusterDecomposition}). Crucially, a final cluster $C\subseteq\{1,\dots,n\}$ may, or may not, contain a collision singularity. Each final cluster evolves as a subsystem
\begin{equation}\label{eq:ClusterMotion}
    \dot q_j=\frac{p_j}{m_j},\qquad
    \dot p_j=-\partial_{q_j}V_C(q)+F_j^C(t),\qquad j\in C,
\end{equation}
where the effective force $F_j^C$ is bounded and $V_C(q)$ is the internal potential energy of the subsystem (see \cref{prop:clustersAreBoundedForce}). Only boundedness is required, so the forces need not have limits. Hence, this framework allows analysis of a colliding subsystem even when another cluster undergoes a simultaneous non-collision singularity, a regime not covered by treating the full orbit as a collision.

Through the final cluster decomposition, we recover, as special cases, the classical Painlev\'e and von Zeipel results as well as the collision asymptotics, while extending them to final clusters in systems for which the interaction need not be attractive and the potential need not diverge on approach to the collision set (see \cref{thm:vonZeipel,thm:PainleveSubSystems}). A key ingredient is a new proof of Sperling's subsystem results based on Graf cluster decompositions \cite{knauf2018mathematical}, that is valid for more general potentials. Ultimately, we use the generalised von Zeipel theorem to establish in \cref{thm:EnergyBound} a bound on the internal energy of any cluster undergoing collision.

From the bound on the internal energy, we are able to use a McGehee blow-up to obtain the collision asymptotics of any subsystem geometrically (see \cref{prop:OmegaLimitSet}). Importantly, it becomes apparent that for non-attractive potentials the normalised configuration of a collision orbit may approach a secondary collision as the singularity time is neared. Regardless, we demonstrate that the moment of inertia obtains the expected scaling law from classical theory (see \cref{cor:DynamicalEstimates}).

As an application of the asymptotic estimates derived in \cref{cor:DynamicalEstimates} of this paper, we use the Poincar\'e-surface approach to the improbability of collision singularities. Following the strategy of Fleischer and Knauf \cite{fleischer2018improbability,fleischer2019improbability}, we construct using \cref{prop:SetFiniteVolume,prop:DefinitionPoincareSurfaceOneSubsystem} a sequence of transverse hypersurfaces that every sufficiently late collision orbit must cross and show that the induced volumes of the surfaces tend to zero. The provided hypersurfaces are simpler than those in prior works. Importantly, the collision asymptotics obtained control the relative positions and momenta on the surfaces, guarantee eventual inward crossing, and provide the decay required for the volume estimate. Consequently, the set of initial conditions giving rise to the collision singularities considered here has phase-space measure zero (see \cref{cor:ImprobabilityCollisions}). 

\section*{Acknowledgment}
We thank the mathematical research institute MATRIX in Australia where the project was started during the workshop \enquote{Nijenhuis Geometry and Integrable Systems II}. M.Q was additionally supported through a MATRIX-Simons Young Scholar Award for a visit to the University of Sydney, during which the mathematical part of the project was finished.

\section{Set-up and Notation}
\subsection{The n-body problem with bounded force}

Consider $n$ interacting particles in $\R^d$. Let $q_i$ be the position in $\R^d$, $m_i$ the mass, and $p_i := m_i \dot{q}_i \in \R^d$ the conjugate momentum of the $i^{\text{th}}$ particle. We consider the motion of the $n$ particles under an interaction given by the potential
\begin{align*}
	V(q) := \sum_{1 \leq i < j \leq n} V_{ij}(q), \qquad V_{ij}  = \frac{Z_{ij}}{\norm{q_i-q_j}^{\alpha}},\qquad Z_{ij}\in \R,\, \alpha \in (0,2)
\end{align*}
with an additional time-dependent external force on each body $F_i$.  The gravitational $n$-body problem is given when $Z_{ij} = - m_i m_j$ and $F_i = 0$.

Denote by $\bN := \{1, \ldots, n\}$ the set of indices of the particles. The potential $V(q)$ is undefined on the set of collisions
\[ \Delta := \{ q = (q_1,\dots, q_n) \in \R^{nd} \,|\, q_i = q_j \text{ for some } i,j\in N, i\neq j \}. \]
Hence, the configuration space is $ M := \R^{nd} \setminus \Delta $ and the phase space is the cotangent bundle
\[ P := T^* M \simeq M \times \R^{nd}.  \]

The dynamics on the phase space $P$ is 
\begin{equation}\label{eqn:MainEquation}
	\begin{aligned}
		\dot{q}_i & = \frac{p_i}{m_i}, \\
		\dot{p}_i & =  F_i-\partial_{q_i} V(q) = F_i + \sum_{j \in \bN \setminus \{i\}}  \frac{\alpha Z_{ij}}{\norm{q_i-q_j}^{\alpha+2}}\left(q_i-q_j\right).
	\end{aligned}
\end{equation}

	\subsection{Cluster Decompositions and some Coordinates}

	To efficiently account for all types of simultaneous collision or non-collision singularities, we use a cluster notation similar to \cite[Definition 2.2, p.5]{knauf2018asymptotic}, compare as well \cite [Chapter 12.6, p.311]{knauf2018mathematical}. 

\begin{definition}[Set partition] \label{def:SetPartition}
	Let $\bN = \{1,\dots, n\}$ be the set of indices of the bodies. A \textbf{cluster decomposition} of $\bN$ is a set $\mathcal{C}:= \{C_1, \ldots, C_k\}$ of subsets $C_i \subseteq \bN$ such that:
	\begin{enumerate}
		\item Each subset $C_i$ is non-empty.
		\item For any $i\neq j$ the subsets $C_i, C_j$ are disjoint.
		\item The union of all $C_i$ is the entire set $\bN$.
	\end{enumerate}
	Each element $C\in\mcal{C}$ of a cluster decomposition is called a \textbf{cluster}.
\end{definition}
In the setting of the $n$-body problem, a cluster decomposition $\C$ corresponds to a division of the particles. In particular, a cluster decomposition could account for a grouping of particles that are close to collision. A particular prominent type of cluster in this work is that of an \emph{isolated} cluster.

\begin{definition}
	For a given solution $q(t)$ on some time domain $t\in(0,T^+)$, a cluster $C$ is said to be \emph{isolated} on $(0,T^+)$ provided that, there exists some $\delta > 0$ such that
	\[  \min_{i\in C,\, j\notin C}\inf_{t\in(0,T^+)} \norm{q_j(t)-q_i(t)} > \delta.  \] 
\end{definition}

Note that of course the whole set of particles is one isolated cluster. So a natural question is, how small an isolated cluster can be? The \textbf{final cluster decomposition} introduced in \cite{quaschner2025improbability} is exactly the splitting of the particles into the finest isolated clusters. Here we give an equivalent definition to that in \cite{quaschner2025improbability}.

\begin{definition}\label{def:FinalClusterDecomposition}
        Define a graph on the particle labels $N=\{1,\dots,n\}$ by joining $i$ and $j$ whenever
        \[ \liminf_{t\uparrow T^+}\|q_i(t)-q_j(t)\|=0.\]
        Then define the \textbf{final cluster decomposition} to be the connected components of this graph. Equivalently, $i$ and $j$ belong to the same final cluster if and only if there is a chain
        $i=i_0,i_1,\dots,i_\ell=j$ such that
        \[ \liminf_{t\uparrow T^+} \|q_{i_{r-1}}(t)-q_{i_r}(t)\|=0,\qquad r=1,\dots,\ell.\]
\end{definition}

Each cluster comes with useful coordinates that splits physical quantities into internal and external quantities.

\begin{definition}[Cluster coordinates] \label{def:ClusterCoordinates} \quad
	\begin{enumerate}[(i)]
		\item For a nonempty subset $C \subseteq \bN$ we define the \textbf{cluster mass} $m_C$ as total mass of the cluster
		      \begin{align*}
			      m_C & := \sum_{i \in C} m_i
		      \end{align*}
		\item For any nonempty subset $C \subset \bN$, we define the \textbf{center of mass of the cluster $C$} and the \textbf{momentum of the cluster $C$} by
			  \[ q_{C} := \frac{1}{m_{C}}\sum_{i \in C} m_i q_i,\qquad p_{C} :=  \sum_{i \in C} p_i.  \]
		\item The \textbf{internal coordinates} and \textbf{internal momenta} of a cluster are given by
			  \[q^I_{C,i} := q_i - q_C,\qquad p^I_{C,i} := p_i - \frac{m_i}{m_C} p_C, \]
		for $i\in C$.
	\end{enumerate}
\end{definition}

\begin{remark}
	Throughout we will often be dealing with a single cluster $C$ of a cluster decomposition. In this situation, we will use the notation $q^I_{i} := q^I_{C,i},\, p^I_{i} := p^I_{C,i}$ for clarity of expressions.
\end{remark}

We now give the notation for various cluster quantities.
\begin{definition}[System quantities] \label{def:System_quantities}
	Let $C \subseteq \{1,\dots,n\}$ be a cluster. 
	\begin{enumerate}[(i)]
		\item The \textbf{moment of inertia}, \textbf{external moment of inertia}, and \textbf{internal moment of inertia} of the cluster $C$ are respectively defined as
		      \begin{align*}
			    	J_C(q) &:=  \frac{1}{2}\sum_{i\in C} m_i \norm{q_i}^2 \\
				   	J_C^E(q) &:= \frac{1}{2} m_C \norm{q_C}^2 \\
					J_C^I(q) &:= \frac{1}{2}\sum_{i \in C} m_i \norm{q_i - q_C}^2 = \frac{1}{2}\sum_{i \in C} m_i \norm{q_{C,i}^I}^2.
		      \end{align*}
		\item The \textbf{kinetic energy}, \textbf{external kinetic energy}, and \textbf{internal kinetic energy} of the cluster $C$ are respectively defined as
		      \begin{align*}
			      	K_C(q,p)  &:=  \sum_{i\in C} \frac{\norm{p_i}^2}{2 m_i} \\
					K_C^E(q,p) & := \frac{\norm{p_C}^2}{2 m_C} \\
					K_{C}^{I}(q,p) & := \sum_{i \in C} \frac{1}{2 m_i}\norm{p_i- \frac{m_i}{m_C}p_C}^2 = \sum_{i\in C} \frac{1}{2 m_i}\norm{p^I_{C,i}}.
		      \end{align*}
		\item The \textbf{potential energy}, \textbf{energy}, and \textbf{internal energy} of the cluster are respectively defined as 
		\begin{align*}
			V_C(q) &:= \sum_{i,j\in C,\, i< j} V_{ij}(q) \\
			h_C(q,p) &:= K_C(p) + V_C(q), \\
			h_C^I(q,p) &:= K_C^I(p) + V_C(q).
		\end{align*}
        \item The \textbf{angular momentum}, \textbf{external angular momentum}, and \textbf{internal angular momentum} of the cluster $C$ are respectively defined as
        \begin{align*}
            L_C(q,p) &:= \sum_{i\in C} q_i \wedge p_i \\
            L_C^E(p,q) &:= q_{C} \wedge p_{C} \\
            L_C^{I}(p,q) &:= \sum_{i \in C} \left(q_i - q_C\right) \wedge \left(p_i- \frac{m_i}{m_C}p_C \right) = \sum_{i\in C} q_{C,i}^I\wedge p_{C,i}^I.
        \end{align*}
	\end{enumerate}
\end{definition}
\begin{remark}
	Often we will be dealing with a single cluster $C$. In which case, the $C$ index may be dropped from the notation when there is no ambiguity.  
\end{remark}

The following lemma is shown in \cite[p.7]{knauf2018asymptotic}, by viewing the internal and external coordinates as orthogonal projections in phase space. However, it also follows by a direct computation. 

\begin{lemma}[Split of system quantities for cluster decomposition] \label{lem:SplitSystemQuantitiesClusterDecomposition} \quad \\
	Let $\mathcal{C}$ be a cluster decomposition. For any $C\in \mcal{C}$ it holds that 
	\[ K_C = K_C^I + K_C^E,\quad J_C = J_C^I + J_C^E,\quad L_C = L_C^I + L_C^E \]
	Let $h$ be the total energy of the system. Then it holds that 
	\[ h_C = h_C^I + K_C^E,\qquad h = \sum_{C\in\C} h_C + h_{\C}^E \]
	where $h_{\C}^E$ is the contribution to the energy coming from potentials between clusters, namely,
	\begin{equation*}
		h_{\C}^E := \sum_{C_i,C_j\in \C,\, i<j} \sum_{(l,m)\in C_i\times C_j} V_{lm}(q)
	\end{equation*}
\end{lemma}

Throughout we will require many estimates and asymptotic properties of various cluster quantities. To easily notate various asymptotic claims, we use the following notation.
\begin{definition}\label{def:AsymptoticNotation}
Consider any two functions $f,g:(0,T) \to \R$ with some fixed $T\in(0,\infty]$.
\begin{enumerate}
	\item If $\lim_{t\to T} f(t)/ g(t) = 1$ then we write $f(t) \sim g(t)$.
	\item If there is a constant $\const > 0$ such that $f(t) \leq \const g(t)$ for every sufficiently large $t$, we write $f(t) \lesssim g(t)$. If in addition it is noteworthy to acknowledge the dependence of $\const$ on a parameter (say $m$) then we write $f(t) \lesssim_{m} g(t)$.
	\item If there exists $0 < c_1 < c_2 < \infty$ such that
    \[ c_1 g(t) \leq f(t) \leq c_2 g(t) \]
    for sufficiently large $t$, 
    then we write $f(t) \asymp g(t)$.
	\item If for every $\epsilon > 0$ it holds that 
	\[ f(t) \leq \epsilon g(t) \]
	for sufficiently large $t$, then we write $f(t) = o(g(t))$.
\end{enumerate}
\end{definition}

Finally, we establish a kind of stability property for the structure of the differential equations of $n$-body problems of bounded force when we consider only the dynamics of an isolated cluster.

\begin{proposition}\label{prop:clustersAreBoundedForce}
	Suppose that $q(t)$ is a solution on $t\in(0,T^+)$ to \cref{eqn:MainEquation} and that $C\subseteq \bN$ is an isolated cluster on $(0,T^+)$. Then, for all $i\in C$,
	\[ \dot{q}_i = \frac{p_i}{m_i},\qquad \dot{p}_i = -\partial_{q_i} V_C(q) + F^C_i  \]
	where $F^C_i(t)$ is bounded on $(0,T^+)$. Moreover,
	\[ \dot{q}^I_i = \frac{p^I_i}{m_i},\qquad \dot{p}^I_i = - \partial_{q^I} V_C(q_i^I) + \tilde{F}^C_i \]
	again with $\tilde{F}^C_i(t)$ bounded on $(0,T^+)$. 

	Additionally, if $T^{+}<\infty$ the center of mass of the cluster $C$ has a definite limit, i.e. $\lim\limits_{t \uparrow T^{+}} q_C(t)$ exists. This last statement does not depend on the $-\alpha$ homogeneity of the potential $V_C$, but only on translation invariance of the potential.
\end{proposition}

\begin{proof}
	The proof is a straightforward calculation. For all $i\in C$ we have 
	\[ \dot{p}_i = -\partial_{q_i} V(q) + F_i = - \partial_{q_i} V_C(q) + \sum_{j\notin C}\alpha \frac{Z_{ij}}{\norm{q_i-q_j}^{\alpha+2}}(q_i-q_j) + F_i. \]
	As the cluster $C$ is isolated, we have $\norm{q_i-q_j} > \delta $ on $t\in (0,T^+)$, hence 
	\[\norm{\sum_{j\notin C}\alpha \frac{Z_{ij}}{\norm{q_i-q_j}^{\alpha+2}}(q_i-q_j)} \leq \frac{\const}{\delta^{\alpha+1}}.\]
	It follows that $\dot{p}_i = -\partial_{q_i} V_C(q) + F^C_i$ for a bounded force $F^C_i$ as claimed.

	Then, using the fact that 
	\begin{align*}
		\dot{q}_C &= \sum_{i\in C} m_i \dot{q}_i = \sum_{i\in C} p_i = p_C \\
		\dot{p}_C &= \sum_{i\in C}\dot{p}_i = \sum_{i\in C} -\partial_{q_i} V_C(q) + F^C_i = \sum_{i\in C} F^C_i
	\end{align*}
	we have
	\begin{align*}
		\dot{q}^I_i &= \dot{q}_i - \dot{q}_C = \frac{p_i}{m_i} - \frac{p_C}{m_C} = \frac{p^I_i}{m_i} 
	\end{align*}
	and
	\begin{align*}
		\dot{p}^I_i &= \dot{p}_i - \frac{m_i}{m_C}\dot{p}_C = -\partial_{q_i^I} V_C(q^I) + F^C_i - \frac{m_i}{m_C} \sum_{i\in C} F_i^C =  -\partial_{q_i^I} V_C(q^I) + \tilde{F}^C_i  
	\end{align*}
	with $\tilde{F}^C_i := F^C_i - \frac{m_i}{m_C} \sum_{i\in C} F_i^C$ bounded.
    
	By the calculation above, $\dot{p}_C$ is bounded, hence $p_C$ has a limit for $t \uparrow T^+$ and thus also $q_C$ converges.
\end{proof}

\section{Existence of limits and bounds of cluster quantities}
In this section we generalise many foundational results concerning the limiting behaviour of system quantities to corresponding cluster quantities. In particular, the theorems of Von Zeipel \cite{von1908singularites,sperling1970real} and Painlev\'e \cite{painleve1897leccons} are generalised, and a bound on change of internal energy $h^I_C$ of a cluster undergoing a collision singularity is proved.

\subsection{Preliminary Lemmas}
We begin by proving some preliminary lemmas giving estimates for the cluster quantities even in the presence of a non-collision or collision singularity. 

Begin with an arbitrary cluster of bodies $C \subseteq \bN$ with dynamics given by \cref{eq:Motion}. As shown in \cref{prop:clustersAreBoundedForce}, the dynamics in $\R^d$ is given by, for each $i\in C$,
\begin{equation}\label{eqn:ClusterEquation}
	\dot{q}_j = \frac{p_j}{m_j},\qquad \dot{p}_j = -\partial_{q_j} V_C(q) + F_j^C(t), 
\end{equation}
where $F_j^C(t)$ is bounded on $(0,T^+)$. In what follows, we set $F_j = F_j^C$ for ease of notation.

\begin{lemma}\label{lemma:hKEstimates}
Assume there exists a solution $(q_i(t),p_i(t))$ to \cref{eqn:ClusterEquation} on $t\in(0,T^+)$. 
Then there are constants depending only on the masses, the force bounds, the escape time $T^+$, and initial conditions such that, for all $t\in[0,T^+)$,
\begin{align}
|h_C^I(t)| &\leq |h_C^I(0)| + \const\int_0^t \sqrt{K_C^I(s)}\,ds, \\
\int_0^t K_C^I(s)\,ds &\leq \frac{2}{2-\alpha} \dot J_C^I(t) + \const. \label{eq:IntegralKineticEnergy}
\end{align}
\end{lemma}
\begin{proof}
	A computation shows that 
	\begin{align*}
		\dot{h}_C^I &= \sum_{i \in C} \frac{\left\langle p_i^I, F_i\right\rangle}{m_i}.
	\end{align*}
	Using Cauchy-Schwarz, we have that 
	\[ |\dot{h}^I_C| \leq  \left( \sum_{i\in C}\frac{\norm{p_i^I}^2}{m_i}\right)^{1/2} \left(\sum_{i\in C}\frac{\norm{F_i}^2}{m_i}\right)^{1/2} \lesssim  \sqrt{2 K^I_C}\]
	 
	The first inequality of the lemma follows from integration of this inequality. 

	For the second inequality, we employ a standard trick. First, observe that 
	\begin{align*}
		\dot{J}^I_C &= \sum_{i\in C} \langle q_i^I, p_i^I \rangle \\ 
		\ddot{J}^I_C &= \sum_{i\in C} \frac{1}{m_i} \norm{p_i^I}^2 -  \left\langle q^I_i, \partial_{q^I_i} V(q^I_i)\right\rangle  +  \left\langle q_i^I, F_i - \frac{m_i}{m_C}\sum_{j\in C} F_j  \right\rangle \\
		 &= 2 K^I_C +\alpha V_C + \sum_{i\in C}\left\langle q_i^I, F_i  \right\rangle \\
		 &= (2-\alpha) K^I_C + \alpha h^I_C + \sum_{i\in C}\left\langle q_i^I, F_i  \right\rangle,
	\end{align*}
	where we have used the fact that $\sum_{i \in C} m_iq_i^I = 0$ to simplify $\ddot{J}^I_C$. Using again Cauchy-Schwarz we have
	\begin{align*}
		\left|\sum_{i\in C}\left\langle q_i^I, F_i  \right\rangle\right| &\leq \left( \sum_{i\in C} m_i\norm{q_i^I}^2 \right)^{1/2} \left(\sum_{i\in C} \frac{\norm{F_i }^2}{m_i}\right)^{1/2}  = \const \sqrt{J^I_C}.
	\end{align*}
	Rearranging the expression for $\ddot{J}^I_C$ for $K^I_C$ and using the estimate yields 
	\[ K_C^I \leq \frac{1}{2-\alpha}\left(\ddot{J}_C^I - \alpha h_C^I\right) + \const \sqrt{J^I_C}. \]
	Integrating up to some time $\tau < T^+$ yields
	\begin{align*}
		\int_0^\tau K^I_C(t)\, dt &\leq \frac{1}{2-\alpha}\left(\dot{J}^I_C(\tau) - \dot{J}^I_C(0) - \alpha \int_0^\tau h^I_C(t)\,dt\right) + \const \int_0^\tau \sqrt{J^I_C(t)}\,dt.
	\end{align*}
	Next, we look for an estimate of the integral over the energy and over the square root of the internal momement of inertia. From the first inequality of the lemma it follows that
	\begin{align*}
		\int_0^\tau |h^I_C(t)|\,dt &\leq  |h^I_C(0)|\tau + \const \int_0^\tau \int_0^t \sqrt{K^I_C(s)}\,ds \\
		&\leq |h^I_C(0)|\tau + \const\cdot \tau \int_0^\tau \sqrt{K^I_C(s)}\,ds.
	\end{align*}

	To bound the integral of the internal moment of inertia, first use Cauchy-Schwarz obtain 
	\[ |\dot{J}^I_C| = \sum_{i\in C} \langle q^I_i, p^I_i \rangle \lesssim \sqrt{K^I_C}{\sqrt{J^I_C}}. \]
	Rearranging yields 
	\[ \frac{d}{dt}\sqrt{J^I_C(t)} \lesssim \sqrt{K^I_C} \]
	Hence
	\[ \int_0^\tau \sqrt{J^I_C(t)}\,dt \leq \tau \sqrt{J^I_C(0)} + \const\cdot\tau \int_0^\tau \sqrt{K_C^I(s)}\, ds \]

	Combining the estimates for the integrals of energy and square root of the moment of inertia in the estimate for the integral of the internal kinetic energy yields 
	\[\int_0^\tau K_C^I(t)\,dt \leq \frac{1}{2-\alpha} \dot{J}^I_C(\tau) + \const + \const\tau \int_0^\tau \sqrt{K^I_C(t)}\,dt \] 

	Now, fix some value $D > 0$ yet to be determined. Observe that for any value $x > 0$, we have $\sqrt{x} \leq \sqrt{D} + \frac{x}{\sqrt{D}}$. Indeed, if $x \leq D$ then $\sqrt{x}\leq \sqrt{D}$, and if $x \geq D$ then $\sqrt{x} \leq x/\sqrt{D}$. Either way, $\sqrt{x} \leq \sqrt{D} + \frac{x}{\sqrt{D}}$. Apply this with $x = K^I_C(t)$ yields 
	\[ \sqrt{K_C^I(t)} \leq \sqrt{D} + \frac{K_C^I(t)}{\sqrt{D}} \]
	for any $D >0$. Then 
	\[ \int_0^\tau K_C^I(t)\,dt \leq \frac{1}{2-\alpha} \dot{J}^I_C(\tau) + \const + \const \tau \left( \tau \sqrt{D} + \frac{1}{\sqrt{D}}   \int_0^\tau K^I_C(t)\,dt\right) \]
	Finally, choosing $D$ large enough so that $\frac{\const\cdot T^+}{\sqrt{D}} < \frac{1}{2}$ and using $\tau < T^+$, we obtain 
	\[ \int_0^\tau K_C^I(t)\,dt \leq \frac{2}{2-\alpha} \dot{J}^I_C(\tau) + \const + \const \tau^2 \sqrt{D} = \frac{2}{2-\alpha} \dot{J}^I_C(\tau) + \const.\]
\end{proof}

We use the above lemma to show that for all clusters experiencing a singularity at $T^+$, the internal moment of inertia $J^I_C$ has a definite limit at the singularity time.
\begin{lemma} \label{lem:ExistenceLimitJCI}
	Let $T^+$ be the singularity time. Then $\lim_{t\uparrow T^+} J^I_C(t)$ exists in $[0,\infty]$.
\end{lemma}
\begin{proof}
	The proof is by contradiction. Assume that 
	\[ a:= \liminf_{t\uparrow T^+} J_C^I(t) < b < \limsup_{t\uparrow T^+} J_C^I(t).  \]
	These limits are only possible provided there are intervals $[t_1,t_2]$ with $t_1$ arbitrarily close to $T^+$ such that $J^I_C$ decreases from $b-\epsilon$ to $a+\epsilon$ for arbitrarily small $\epsilon$. By the mean value theorem, there exists $t_* \in (t_1,t_2)$ such that
	\[ \dot{J}^I_C(t_*) = -\frac{b-a + 2\epsilon}{t_2 - t_1}. \]
	The interval $[t_1, t_2]$ can be taken arbitrarily close to time $T^+$, hence, for any value $M_* > 0$, we can choose $\Delta t = t_2 - t_1$ to be arbitrarily small so that 
	\[ \dot{J}^I_C(t_*) \leq -M_*. \]
	But the inequality \[ \int_0^\tau K_C^I(t)\, dt \leq \frac{2}{2-\alpha} \dot{J}^I_C(\tau) + \const \] proven in \cref{lemma:hKEstimates}, together with the fact that $K^I_C(t)$ is always non-negative, shows that $\dot{J}^I_C$ is bounded from below. A clear contradiction. Hence, the limit $\lim_{t\uparrow T^+} J^I_C(t)$ must exist on $[0,\infty]$.
\end{proof}

From the above calculations, we obtain the following corollaries:

\begin{corollary} \quad
	\begin{enumerate}[(i)]
		\item If $\liminf_{t \uparrow T^{+}} \dot{J}_C^{I}(t) < \infty$, then $\int_0^{T^{+}} K_C^I(s)\,ds < \infty$.
		\item $\liminf_{t \uparrow T^{+}} \dot{J}_C^{I}(t) > -\infty$.
		\item If $\limsup_{t \uparrow T^{+}} \dot{J}_C^{I}(t) = \infty$ and $\limsup_{t \uparrow T^{+}} J_C^{I}(t) < \infty$, then  $\lim_{t \uparrow T^{+}} \dot{J}_C^{I}(t) = \infty$
	\end{enumerate}
\end{corollary}

\begin{proof}
	(i) and (ii) are immediate from \eqref{eq:IntegralKineticEnergy}. \\
	Proving (iii) is a little more involved. Assume to the contrary that $\liminf_{t \uparrow T^{+}} \dot{J}_C^{I}(t) < \infty$. Using (i) we obtain that $\int_0^{T^{+}} K_C^I(s)\,ds < \infty$ and hence also $h_C^{I}$ is bounded. Denote some global bound for $J_C^{I}(t)$ by $A>0$. \\
	Our idea is to prove that if $\dot{J}_C^{I}$ is large at some time, it will be growing monotonically and hence cannot decrease again. This contradicts the assumption that the $\liminf$ is finite. For this just note that
	\begin{align*}
		\ddot{J}^I_C &= (2-\alpha) K^I_C + \alpha h^I_C + \sum_{i\in C}\left\langle q_i^I, F_i  \right\rangle \geq (2-\alpha) K^I_C - \const, \\
		K^I_C &\geq \frac{\left(\dot{J}_C^{I}\right)^2}{J_C^{I}} \geq \frac{\left(\dot{J}_C^{I}\right)^2}{A}.
	\end{align*}
	So as soon as $\dot{J}^I_C(t_0)$ is sufficiently large, $	\ddot{J}^I_C>0$ and thus the claim follows.
\end{proof}

\begin{remark}
	Case (iii) can actually not occur. This an immediate consequence of the generalisation of the theorem of von Zeipel that we prove in \cref{thm:vonZeipel}.
\end{remark}

The above lemmas give useful bounds on the change of the energy and the integral over the kinetic energy for a cluster $C$. However, we cannot yet conclude that these quantities are bounded (which is even false for orbits that lead to non-collision singularities). To do this, we will need a generalisation of Von Zeipel's theorem and a theorem of Painlev\'e.

\subsection{Generalisation of Von Zeipel's Theorem}
    \label{sec:vonZeipel}
	As established in \cref{lem:ExistenceLimitJCI}, there exists a limit of the internal moment of intertia $J_C^I(t)$ as $t\to T^+$ for any cluster $C$. The limit of $J_C^I(t)$ can take values on $[0,\infty]$. There is a natural dichotomy that arises from the posible values of the limit. The first case is when $\lim\limits_{t \uparrow T^{+}} J_C^{I}(t)=\infty$ and corresponds to a so-called noncollision singularity. The second when $\lim\limits_{t \uparrow T^{+}} J_C^{I}(t) < \infty$. The purpose of this section is to show that, when the limit of $J^I_C(t)$ is finite, the motion of the bodies $q_i$ in the cluster $C$ is very well-behaved; their pairwise distances have a definitive limit as $t\to T^+$. Moreover, if $C$ is a cluster in the final cluster decomposition (see \cref{def:FinalClusterDecomposition}) containing more than one body, then the limit of $J^I_C(t)$ is 0 and $C$ admits a total collision. 
	This is essentially Von Zeipel's theorem, but in the more general setting of bounded forces. 

	We will prove the generalisation of von Zeipel's theorem under the following general assumptions.
	    \begin{remark}[More general assumptions for the proof of von Zeipel] \label{rem:GeneralAssumptions} 
			~
			\begin{enumerate}
            \item The motions of the particles is given by \begin{equation}
	           \dot{q}_j = \frac{p_j}{m_j},\quad \dot{p}_j = -\sum_{k\in C} \frac{\partial V_{jk}}{\partial q_j} + F_j(t), 
            \end{equation}
            where the external forces $F_j$ are bounded and for each $\epsilon>0$ the potentials $V_{jk}$ and their derivatives $\frac{\partial V_{jk}}{\partial q_j}$ are bounded on $R^{n}\setminus B_{\epsilon}(0)$ and translation invariant.
            \item The limit $\lim\limits_{t \uparrow T^{+}} J_C^{I}(t)$ exists in $[0,\infty)$ for this particular trajectory.
        \end{enumerate}
        Note that these assumptions are satisfied for our homogeneous potentials of degree $-\alpha$ by \cref{prop:clustersAreBoundedForce,lem:ExistenceLimitJCI}.
    \end{remark}

	\begin{theorem}[von Zeipel] \label{thm:vonZeipel} \quad \\
		Let $C$ be an isolated cluster near the finite time $T^+$. If $J_C^I(t)$ has a finite limit as $t \uparrow T^+$, then the internal configuration $q^I_C(t)$ has a finite limit. If $T^+$ is a singular time for the internal dynamics of $C$, then the limiting internal configuration belongs to the collision set. 
	\end{theorem}
	
	Before we prove the theorem, we give a corollary stating some dichotomy for clusters in the final cluster decomposition.

	\begin{definition}[Collision singularity] \quad \\
		For an isolated cluster $C$ with $\left|C\right|\geq 2$ we say that the particles of $C$ have a collision singularity at time $T^{+}$ if, for all $i,j \in C$ with $i \neq j$, the limits $\lim\limits_{t \uparrow T^{+}} \norm{q_i(t) - q_j(t)}$ exist, and at least for one pair $i\neq j \in C$ we have $\lim\limits_{t\uparrow T^{+}} \norm{q_i(t) - q_j(t)} = 0$. If for all $i,j\in C$ it holds that $\lim_{t\uparrow T^+} \norm{q_i(t) - q_j(t)} = 0$ then $C$ is said to have a total collision singularity.
	\end{definition}

    \begin{corollary}[von Zeipel for final clusters]\label{cor:vonZeipelFinalClusters}\quad \\
        Under the assumptions of \cref{thm:vonZeipel}, for each cluster $C$ in the \textbf{final} cluster decomposition, either:
		\begin{enumerate}
			\item $\lim\limits_{t \uparrow T^{+}} J_C^{I}(t) = 0$ and there is a total collision of the cluster $\lim\limits_{t \uparrow T^{+}} q_i^{I} = 0$, or,
			\item $\lim\limits_{t \uparrow T^{+}} J_C^{I}(t) = \infty$ and there is not a collision singularity (the cluster is said to have a non-collision singularity).
		\end{enumerate}
    \end{corollary}

    \begin{proof}
        If there is a total collision singularity then clearly the internal moment of inertia $J^I_C(t)$ has a limit to $0$. Conversely, assume that $\lim\limits_{t \uparrow T^{+}} J_C^{I}(t) < \infty$. By \cref{thm:vonZeipel} this implies the existence of $\norm{q_i-q_j}$ for all $i,j \in C$. But as $C$ is a final cluster, for each $i_0$ there is an index $j_0$ with $\liminf_{t \uparrow T^{+}} \norm{q_{i_0}-q_{j_0}} = 0$ and hence also $\lim_{t \uparrow T^{+}} \norm{q_{i_0}(t)-q_{j_0}(t)} = 0$. So we can group the particles by the relation $i \sim j \quad \Leftrightarrow \quad \lim_{t \uparrow T^{+}} \norm{q_{i}(t)-q_{j}(t)} = 0$. But by the definition of the final cluster decomposition, all particles of $C$ have to be in the same cluster, which implies convergence of all particles to the center of mass of $C$.
    \end{proof}
	
	For the proof of \cref{thm:vonZeipel}, we start by introducing some notation. Relabeling the particles of $C$, we assume without loss of generality that $C$ is given by $C= \{1, \ldots, n\} $ for some $n \in \N$. In the proof, we only use the assumption on the dynamics and potential given in \cref{rem:GeneralAssumptions}. 
	
	The first trick from Sperling \cite{sperling1972collision} is to observe that the assumed finite limit of $J^I_C(t)$ as $t\uparrow T^+$ forces only four possibilities for the dynamics between two bodies $q_i$ and $q_j$ of each pair in $C$. Firstly, as 
	\[ J^I(t) = \frac{1}{2 m_C} \sum_{i < j} m_i m_j \norm{q_i(t) - q_j(t)}^2, \]
	the finite limit of $J^I(t)$ implies that the pairwise distances $\norm{q_i(t) - q_j(t)}$ for all pairs $\{i,j\}$ in $C$ are uniformly bounded. Consequently, the only four possibilities for a pair $\{i,j\}$ in $C$ are:
	\begin{enumerate}
		\item The pair $\{i,j\}$ end in a collision. That is, $\{i,j\}$ is part of the set 
		\[ G_1 := \{\{i,j\} \subseteq C \mid \lim_{t \uparrow T^{+}} \norm{q_i(t) - q_j(t)} = 0 \}. \]
		\item The pair $\{i,j\}$ has a well-defined limiting distance. That is, $\{i,j\}$ is part of the set
		\[ G_2 := \{\{i,j\} \subseteq C \mid \lim_{t \uparrow T^{+}} \norm{q_i(t) - q_j(t)} >0 \text{ and the limit exists } \}. \]
		\item The pair $\{i,j\}$ does not have a limiting distance, however, the particles remain away from collision. That is, $\{i,j\}$ is part of the set 
		\[G_3 := \{\{i,j\} \subseteq C \mid 0 <\liminf_{t \uparrow T^{+}} \norm{q_i(t) - q_j(t)}< \limsup_{t \uparrow T^{+}} \norm{q_i(t) - q_j(t)} < \infty \}.\]
		\item The pair $\{i,j\}$ do not have a limiting distance and they stay arbitrarily close to collision. That is, $\{i,j\}$ is part of the set
		\[ G_4 := \{\{i,j\} \subseteq C \mid 0 =\liminf_{t \uparrow T^{+}} \norm{q_i(t) - q_j(t)}< \limsup_{t \uparrow T^{+}} \norm{q_i(t) - q_j(t)} < \infty \} .  \]
	\end{enumerate}
	To prove \cref{thm:vonZeipel}, it suffices to show that $G_3 = \emptyset = G_4$, so that every mutual distance has a finite limit. If the internal dynamics of $C$ is singular, then at least one pair belongs to the set $G_1$. Before the formal proof that there is no $G_3$ or $G_4$ pairs, we want to give an informal analogy as to why they cannot occur with a finite limit for the moment of inertia.
	
	\begin{remark} \quad \\
		We want to give an informal argument what is the problem with oscillatory motion under convergence of $J(t)$ to a finite limit. To give the reader a better picture for the argument, one can compare it to the (unbounded) motion in the five-body problem constructed by Xia \cite{xia1992existence}. In this example, a so-called messenger particle is commuting between two binary clusters, having near-collisions with them and pushing the system to infinity via this oscillation. For this, there is a natural clustering: the binaries each form a cluster and while the messenger particle in the middle is traveling, it is a cluster on its own. However, as soon as the messenger approaches one of the binaries, one can consider them as a cluster of three. In the language of the Graf partition defined in \cref{def:GrafPartition}, this is then a time where the Graf partition only consists of two clusters, which is the minimal number in this case. Now as the messenger separates again from the binary, their distance will have again a contribution to the moment of inertia. For unbounded motions as in Xia's example this is no problem (the moment of inertia just increases, as it does all the time along the trajectory), but what happens if the moment of inertia should have a definite limit? In this case, the increase of this part of the moment of inertia has to be compensated by a decrease of the external moment of inertia given by the centers of the clusters. This implies the existence of a local maximum of the external moment of inertia. But then (as the clusters are separated) the kinetic energy cannot be large, which however also contradicts that a corresponding decrease of the external moment of inertia could happen in this short time interval. This is the technical core of \cref{prop:G4empty}.
	\end{remark}
	
	To prove \cref{thm:vonZeipel}, we will make use of Graf partitions of configuration space, indexed by cluster decompositions (see e.g. \cite[Section 12.6]{knauf2018asymptotic}). In essence, a Graf partition splits configuration space up into regions where the bodies are grouped into clusters with internal moments of inertia of comparable size. While there is some initial set-up of various definitions and properties of Graf partitions, once this is done, they provide a convenient way to understand the various asymptotic behaviours of pairs of particles in $C$. 

	\begin{definition}[Graf partition] \label{def:GrafPartition} \quad \\
		For $\delta \in (0,1)$, let
		\begin{align*}
			J^{(\delta)}:& \; M \rightarrow [0, \infty), \quad J^{(\delta)}(q) := \max\left\{ J_{\C}^{E}(q) + \delta^{\left|\C \right|} \mid \C \in \mathcal{P}(\bN)\right\}
		\end{align*}
		where 
		\begin{equation}\label{eq:JCE}
			J_{\C}^{E}(q) := \sum_{C \in \C} \frac{m_C}{2} \norm{q_C}^2.
		\end{equation}
		The function $J^{(\delta)}$ induces subsets of $M$ given by
		\begin{align*}
			\Xi^{(\delta)}_{\C} &:= \left\{ q \in M \mid J^{(\delta)}(q) = J_{\C}^{E}(q) + \delta^{\left|\C \right|} \right\}
		\end{align*}
		for arbitrary $\C \in \mathcal{P}(\bN)$. The family $\left(\Xi_{\C}^{(\delta)}\right)_{\C \in \mathcal{P}(\bN)}$ is a finite closed cover of $M$ and is a partition of $M$ modulo sets of Lebesgue measure zero. It is denoted the \textbf{Graf partition} of $M$. For a position $q \in M$ any $\C \in \mathcal{P}(\bN)$ is a \textbf{Graf cluster decomposition} of $q$ if $q \in \Xi_{\C}^{(\delta)}$.
	\end{definition}

    Before we state some properties of the Graf cluster decomposition, we have to recall the partial order on the set of cluster decompositions as introduced e.g. in \cite{knauf2018asymptotic}.

    \begin{definition}[Partial order for cluster decompositions] \label{def:PartialOrderClusterDecomposition}\quad \\
        The set $\mathcal{P}(\bN) := \{ \mathcal{C} \mid \mathcal{C} \text{ is a set partition of } \bN \}$ becomes a lattice with the partial order defined by the refinement relation
			\begin{align*}
			\mathcal{C} = \{C_1, \ldots, C_k \} &\preccurlyeq \{D_1, \ldots, D_l \} = \mathcal{D}
			\intertext{if there is a surjection $\pi: \{1, \ldots, k\} \rightarrow \{1, \ldots, l\}$ with}
			C_{i} &\subseteq D_{\pi(i)}, \quad \forall i \in \{1, \ldots, k\}.
			\end{align*}
			In this case, $\mathcal{C}$ is called \textbf{finer} than $\mathcal{D}$ and $\mathcal{D}$ is called \textbf{coarser} than $\mathcal{C}$.
			The finest cluster decomposition is $\C_{\min} = \{\{1\}, \ldots, \{n\}\}$.
    \end{definition}
	
	\begin{remark} \label{rem:PropertiesGrafPartition}
		\quad \\
		We gather some useful statements from \cite[Chapter 12.6]{knauf2018mathematical} on the Graf partition. Henceforth, we will assume that $ 0 < \delta < \delta_0$, for some $\delta_0$ sufficiently small that the following properties hold.
		\begin{enumerate}[(i)]
			\item The Graf partition is only a measure-theoretical partition of the phase space in the sense that $\Xi_{\C}^{(\delta)} \cap \Xi_{\D}^{(\delta)} \neq \emptyset$ is possible, but the intersection is always a set of Lebesgue measure zero if $\C\neq \D$. Alternatively stated, the set of $q$ with more than one Graf cluster decomposition is Lebesgue measure zero.
			\item If two bodies are in the same cluster of a Graf cluster decomposition then they have at most some maximal distance $\rho^I(\delta)$ such that $\rho^I(\delta) \to 0$ as $\delta \to 0$.
			\item \label{bul:GrafPartitionMinimalDistance} For sufficiently small $\delta$, if two bodies are not in the same cluster, they have at least a distance $\rho^E(\delta)$.
			\item \label{bul:DeltaGrafPartitionSufficientlySmall} For sufficiently small $\delta$, if for some position $q$ there are two distinct Graf cluster decompositions $\C$ and $\D$, then these two decompositions are comparable in the partial order of Definition~\ref{def:PartialOrderClusterDecomposition}.
		\end{enumerate}
	\end{remark}

	Define now an equivalence relation on $\bN$ by 
	\[ i \sim_0 j \iff i=j \text{ or } \{i,j\}\in G_1. \]
	Transitivity follows from the triangle inequality. Let
	\[ \D^{(0)} := \bN / \sim_0 \]
	be the corresponding cluster decomposition. The clusters of $\D^{(0)}$ consist of bodies whose mutual distances tend to zero. Distinct clusters need not yet be isolated, since a pair of bodies belonging to different clusters may lie in $G_4$. To describe the time-dependent clustering before $G_4$ has been excluded, we use the temporal Graf cluster decomposition from \cite[Definition 2.12]{quaschner2025improbability}.

	\begin{definition}[Temporal Graf cluster decomposition] \label{def:TempAndFinalClusterDecomposition} \quad \\
		Let $T^{+} > 0$ be an arbitrary time and $q: [0, T^{+}) \rightarrow M$ be a continuous function.
		\\
		For arbitrary $\delta \in (0,1)$ we call a family $\left(\C_t^{(\delta)}\right)_{t \in [0, T^{+})}$ with $\C_t^{(\delta)} \in \mathcal{P}(\bN)$ a \textbf{temporal Graf cluster decomposition}, if
		\begin{align*}
			q(t) \in \Xi_{\C_t^{(\delta)}}^{(\delta)} \quad \forall t \in [0, T^{+}).
		\end{align*}
	\end{definition}
	
	\begin{remark}
		A temporal Graf cluster decomposition is not unique at times where $q(t)$ is in the intersection of two or more possible clusters. For $\delta$ sufficiently small, the possible Graf decompositions at any $q$ are comparable. Consequently, one obtains a canonical temporal Graf decomposition by always selecting either the finest or the coarsest decomposition. 
	\end{remark}
	\begin{remark}
		If there is a cluster decomposition $\C$, some $\delta>0$ and some interval $(t_1, t_2)$ such that $\forall t \in (t_1, t_2):  \C_t^{(\delta)} = \C$, then the clusters of $\C$ are isolated clusters on the interval $(t_1, t_2)$ with a minimal distance $\rho^{E}(\delta)$, hence all forces between the clusters are bounded.
	\end{remark}

	With the temporal Graf partition defined, we are now in a position to prove $G_3 = \emptyset = G_4$, thus establishing \cref{thm:vonZeipel}.	Although pairs in $G_3$ have no limiting distance, their interaction forces remain uniformly bounded because their separation is bounded away from zero. Pairs of $G_4$ have the most complicated asymptotics. The main crux of the proof is to establish that $G_4$ is empty. In fact, $G_4 = \emptyset$ is sufficent to show that $G_3 = \emptyset$.

	\begin{lemma}[Collision singularity for $G_4 = \emptyset$] \label{lem:G4empty} 
		Under the assumptions of \cref{rem:GeneralAssumptions}, if $G_4 = \emptyset$, then $G_3 = \emptyset$ and all limits $\lim\limits_{t \uparrow T^{+}} q_i(t) -q_j(t)$ exist as vectors in $\R^{d}$.
	\end{lemma}
	\begin{proof}
		Take $\delta>0$ sufficiently small so that no pair $\{i,j\}$ in $G_2$ or $G_3$ can belong to the same cluster of a Graf decomposition with parameter $\delta$.
		Now consider a temporal Graf cluster decomposition $\C^{(\delta)}_t$ with this parameter $\delta$. 
		After some time $t_0$, all pairs of particles $\{i,j\} \in G_1$ have a sufficiently small distance (depending on $\delta$) and hence are in the same cluster of $\C_t^{(\delta)}$. Hence, every cluster of $\D^{(0)}$ is contained in a cluster of $\C^{(\delta)}_t$.
		
		Because $G_4 = \emptyset$, every pair of bodies belonging to distinct clusters of $\D^{(0)}$ must lie in either $G_2$ or $G_3$. But by assumption on the size of $\delta$, such a pair cannot belong to the same cluster of $\C_t^{(\delta)}$. Therefore 
		\[ \C_t^{(\delta)} = \D^{(0)}, \text{ for every } t \in (t_0, T^+). \]
		
		It follows that the distinct clusters of $\D^{(0)}$ remain uniformly separated on $(t_0, T^+)$, that is, each of the clusters of $\D^{(0)}$ are isolated clusters. In consequence, $\lim\limits_{t \uparrow T^+} q_D(t)$ exists for all $D\in \D^{(0)}$, compare \cref{prop:clustersAreBoundedForce}.
        As the pairwise distance between particles within such a cluster $D$ converges to $0$, this implies that all particles have a limit for $t \uparrow T^+$ and hence there are no pairs of type $G_3$.
	\end{proof}
	
	Having established \cref{lem:G4empty}, in order to establish \cref{thm:vonZeipel}, it suffices to show $G_4 = \emptyset$. Henceforth, assume $G_4 \neq \emptyset$ and we will prove a contradiction. The key idea is that pairs of type $G_4$ are oscillating, which implies that for each fixed $\delta$ the temporal Graf cluster decomposition will change infinitely often and arbitrary close to the singular time $T^+$. What we want to find now are coarsest cluster decompositions with a minimal number of clusters which recur infinitley often close to the singular time. 
	
	\begin{definition} \quad
		\begin{enumerate}[(i)]
			\item For a fixed $\delta>0$, let $\mathcal{R}_{\ast}^{(\delta)}$ be the set of cluster decompositions that appear infinitely often and arbitrarily close to $T^+$ in a temporal Graf decomposition $\C^{(\delta)}_t$. Any decomposition $\C \in \mathcal{R}_{\ast}^{(\delta)}$ is called \textbf{recurrent at $T^+$}. That is, 
			\begin{align*}
				\mathcal{R}_{\ast}^{(\delta)} &:= \left\{ \C \in \mathcal{P}(\bN) \mid \forall t \in [0, T^+), \;  \exists s \in (t, T^+) \text{ such that } \C_s^{(\delta)} = \C \right\}, 
			\end{align*}
			\item Define 
			\[ k_*(\delta) := \min_{\C \in \mathcal{R}^{(\delta)}_*} |\C|. \]
			Let $\mathcal{R}_{\ast, \min}^{(\delta)}$ be the set of recurrent decompositions at $T^+$ with minimum number of clusters $k_*(\delta)$. That is,
			\[
				\mathcal{R}_{\ast, \min}^{(\delta)} := \left\{ \C \in 	\mathcal{R}_{\ast}^{(\delta)} \mid \left|\C\right| = k_*(\delta) \right\}.
			\]
			\item We define $\mathcal{R}_{\ast, \min}^{(0)} := \lim\limits_{\delta  \downarrow 0} \mathcal{R}_{\ast, \min}^{(\delta)}$.
		\end{enumerate}
	\end{definition}
	
	\begin{remark} \quad
		\begin{enumerate}[(i)]
			\item Every recurrent Graf decomposition needs to be coarser than $\D^{(0)}$. Indeed, for each $\delta>0$, there is some time after which all pairs in $G_1$ must be in the same Graf cluster, as their internal distance are way smaller than any fixed constant depending on $\delta$. 
			\item The recurrent cluster decompositions need not be comparable. But by choosing a recurrent decomposition with minimum number of clusters $k_*(\delta) = |\C|$, the existence of coarser cluster decompositions can be excluded. 
		\end{enumerate}
	\end{remark}
	
	The following lemma establishes properties of $k_*$ and $\mathcal{R}^{(\delta)}_{*,\min}$ as $\delta$ varies.

	\begin{lemma}\label{lem:BehaviourLimitDelta}
		Fix $\bar\delta>0$ sufficiently small so that $\bar\delta\leq 1/n$ and all Graf decompositions are comparable for every $0<\delta\leq\bar\delta$. Let $\delta_1,\delta_2$ be such that $0 < \delta_1 < \delta_2 < \bar{\delta}$. Then:
		\begin{enumerate}[(i)]
			\item For any $q\in M$, if $\C$ is a $\delta_1$-Graf decomposition and $\D$ a $\delta_2$-Graf decomposition at $q$, then $|\D| \leq |\C|$ with strict inequality if $\D \neq \C$.
			\item The function
			$ k_*:(0,\bar\delta] \to \left\{1,\dots,\big|\mathcal D^{(0)}\big|\right\} $
			is non-increasing. Consequently, for each $k$, the level set
			$ I_k:=\{\delta\in(0,\bar\delta]\mid k_*(\delta)=k\} $
			is a (possibly empty) interval.
			\item If $ k_*(\delta_1)=k_*(\delta_2)$, then
			\[ \mathcal{R}_{*,\min}^{(\delta_1)} \subseteq \mathcal {R}_{*,\min}^{(\delta_2)}. \]
			\item There exist $k_0\in\bN$, a cluster decomposition
			$\mathcal C^0$ with $|\mathcal C^0|=k_0$, and numbers
			$\delta_0>0$ and $t_0<T^+$ such that
			\[
			|\mathcal C_t^{(\delta)}|\geq k_0
			\]
			for every $0<\delta\leq\delta_0$ and
			$t\in[t_0,T^+)$, and such that, for every
			$0<\delta\leq\delta_0$ and every $s<T^+$, there exists
			$t\in(s,T^+)$ satisfying $\C_t^{(\delta)} = \C^0$.
		\end{enumerate}
	\end{lemma}

	\begin{proof}
	Suppose that, at some $q$, the decompositions $\C$ and $\D$ are such that $q\in \Xi^{(\delta_1)}_{\C}$ and $q\in \Xi^{(\delta_2)}_{\D}$.
	By maximality,
	\[ J_{\C}^E(q)+\delta_1^{|\C|} \geq J_{\D}^E(q)+\delta_1^{|\D|} \]
	and
	\[ J_{\D}^E(q)+\delta_2^{|\D|} \geq J_{\C}^E(q)+\delta_2^{|\C|}. \]
	Thus
	\[ \delta_1^{|\D|}-\delta_1^{|\C|} \leq J_{\C}^E(q)-J_{\D}^E(q) \leq \delta_2^{|\D|}-\delta_2^{|\C|}. \]
	From this, it can be concluded that $|\D| \leq |\C|$.

	If $|\D| = |\C|$ then we would have $J_{\C}^E(q) = J_{\D}^E(q)$. Consequently, $\C,\D$ are Graf decompositions of $q$ for both the values of $\delta_1, \delta_2$. In turn, this implies that $\C$ and $\D$ are comparable (as $\delta$ is sufficiently small). But, comparable decompositions with the same number of clusters must be equal. Hence, it has been shown that 
	\[ |\D| \leq |\C| \]
	with strict inequality unless $\C = \D$. This establishes (i).

	Now let $ \C\in\mathcal{R}_{*,\min}^{(\delta_1)}$.
	Choose a sequence $s_\ell\uparrow T^+$ such that $ \C_{s_\ell}^{(\delta_1)}=\C$. Since there are finitely many cluster decompositions, after passing
	to a subsequence, we may assume that there is a fixed decomposition $\D$ such that $\C_{s_\ell}^{(\delta_2)}=\D$
	for all $s_{\ell}$. Then, by definition, $\D\in\C_*^{(\delta_2)}$, and consequently
	\[ k_*(\delta_2) \leq |\D|\leq|\C|=k_*(\delta_1).\]
	Therefore $k_*$ is non-increasing in $\delta$. Since $k_*$ is integer-valued and non-increasing, each of its level sets is an interval as claimed in (ii).

	If $k_*(\delta_1)=k_*(\delta_2)$ then necessarily $|\D|=|\C|$. Then, as shown above, $\D=\C$. Thus $\C$ is recurrent at
	parameter $\delta_2$, and hence $\C\in\mathcal{R}_{*,\min}^{(\delta_2)}$. This proves (iii).

	Since $k_*$ is integer-valued and non-decreasing as
	$\delta \to 0$, it is eventually constant. Let
	\[ k_0:=\lim_{\delta\downarrow0}k_*(\delta). \]
	For all sufficiently small $\delta$, we therefore have
	$k_*(\delta)=k_0$. On this interval, part (iii) shows that the
	nonempty finite sets $\mathcal{R}_{*,\min}^{(\delta)}$ form a nested
	family as $\delta \to 0$. Since there are only finitely many
	cluster decompositions, this family eventually stabilizes. Hence there exist $\delta_0>0$ and a decomposition $\C^0$ such
	that $ \C^0\in\mathcal{R}_{*,\min}^{(\delta)} $
	for every $0<\delta\leq\delta_0$. In particular, $|\C^0|=k_0$ and $\C^0$ occurs arbitrarily close to $T^+$ for every
	such $\delta$.

	Finally, since no decomposition with fewer than $k_0$ clusters is recurrent at parameter $\delta_0$, there exists $t_0<T^+$ such that $|\C_t^{(\delta_0)}|\geq k_0$
	for every $t\in[t_0,T^+)$. Hence, 
	$|\C_t^{(\delta)}| \geq |\C_t^{(\delta_0)}|$, whenever $0<\delta\leq\delta_0$. Therefore
	$|\C_t^{(\delta)}|\geq k_0$
	for every $0<\delta\leq\delta_0$ and every
	$t\in[t_0,T^+)$. This establishes (iv).
	\end{proof}

\begin{remark}
The argument below is the Graf-partition analogue of Sperling's separation-scale argument. Sperling introduces a spatial threshold and studies the number of groups of bodies separated at that scale. Here the role of the spatial threshold is played indirectly by the Graf parameter $\delta$. That is, a Graf decomposition has individual sizes of its distinct clusters bounded above by a quantity $\rho^I(\delta)\to 0$ as $\delta \to 0$, while distinct clusters are separated by a positive distance $\rho^E(\delta)$. The number of clusters in a Graf decomposition is therefore the corresponding discrete quantity to Sperling's number of groups of bodies.

As in Sperling's proof, this integer-valued quantity eventually stabilizes as the scale is varied. The final contradiction compares two sufficiently separated Graf scales $\delta_1 \ll \delta_2$. The change of scale produces a fixed gap in the external moment of inertia, whereas bounded forces between clusters imply that this inertia cannot vary by the required amount on a sufficiently short time interval near $T^+$.
\end{remark}

\begin{proposition}\label{prop:G4empty}
Under the assumptions of \cref{rem:GeneralAssumptions}, there are no pairs of type $G_4$.
\end{proposition}

\begin{proof}
Assume that $G_4\neq \emptyset$. Throughout, let $\C^{(\delta)}_t$ be the unique temporal Graf partition with $\C^{(\delta)}_t$ the coarsest Graf partition at $q(t)$, for all $t \in (0,T^+)$. By \cref{lem:BehaviourLimitDelta}, there is $k_0$, $\delta_0>0$, $t_0<T^+$, and $\C \in \C^{(0)}_{*,\min}$ with $|\C| = k_0$ such that, for all $\delta, s$ with $0 < \delta < \delta_0$ and $ t_0 < s < T^+$, it holds that $|\C^{(\delta)}_s| \geq k_0$ and there exists $t\in (s,T^+)$ satisfying $\C_t^{(\delta)} = \C$.

Choose parameters $ 0<\delta_1<\delta_2\leq\delta_0 $ sufficiently small so that all the Graf separation and comparability properties of \cref{rem:PropertiesGrafPartition} hold. Moreover, choose $\delta_1$ so that $\delta_1^{k_0}+\delta_2^{k_0+1} \leq \tfrac12 \delta_2^{k_0}$.
Since $\lim_{t\uparrow T^+}J(t) < \infty$ exists, there is $u_0<T^+$ such that $|J(t_1)-J(t_2)|<\delta_1^n$ for all $t_1, t_2\in[u_0,T^+)$.

We will choose $s\geq\max\{t_0,u_0\}$ to be close to $T^+$ in a precise way below. For now, let it be arbitrary. By the recurrence property of $\C$, there exists $t_p\in(s,T^+)$ such that $ \C_{t_p}^{(\delta_1)}=\C$. We claim that also $\C^{(\delta_2)}_{t_p} = \C$. Indeed, let $\D$ be any Graf decomposition of $q(t_p)$ for $\delta_2$. From \cref{lem:BehaviourLimitDelta} it must hold that $|\D|\leq|\C|=k_0$, with equality only if $\D=\C$. On the other hand, the choice of $t_0$, together with the convention that $\C_t^{(\delta_2)}$ is the coarsest Graf decomposition at $q(t_p)$, rules out any Graf decomposition with fewer than $k_0$ clusters. Hence $\D=\C$.

By continuity, we must have $\C^{(\delta_2)}_t = \C$ on a maximal open interval $(t_1, t_2) \ni t_p$. Without loss of generality, it can be assumed that $t_1 \geq s$, otherwise, a later time $t_p$ can always be chosen due to the recurrence of $\C$ and assumption that $G_4\neq \emptyset$. Indeed, for any pair $\{i,j\}\in G_4$, choose $\delta_2$ sufficiently small so that $\limsup_{t\uparrow T^+}\norm{q_i(t) - q_j(t)} > \rho^I(\delta_2)$, where $\rho^I$ is the minimal distance of particles in a Graf decomposition (recall \cref{rem:PropertiesGrafPartition}). Then there are arbitrarily late times at which $i,j$ are not in the same cluster of the Graf partition. On the other hand, as $\liminf_{t\uparrow T^+} \norm{q_i(t) - q_j(t)} = 0$ there are arbitrarily late times at which they must be in the same cluster of a Graf decomposition. Hence, the maximal interval $(t_1,t_2)$ cannot extend to $t_2 = T^+$, thus $t_1$ can be assumed greater than any chosen $s <T^+$. 

At the end points $q(t_1), q(t_2)$, $\C$ cannot be the only Graf partition. Set $\C_1, \C_2 \neq \C$ as distinct Graf partions at $q(t_1), q(t_2)$, respectively.
By assumption on the size of $\delta_0$, the decompositions $\C_i$ and $\C$ are comparable. As $\C$ is the coarsest choice of Graf partition and $\C_i \neq \C$, we must have $|\C| < |\C_i|$, so $|\C_i| \geq k_0+1$, from \cref{lem:BehaviourLimitDelta}.

We now claim that $J^E_{\C}(t_i) < J^E_{\C}(t_p)$ and consequently, $J^E_{\C}(t)$ obtains a maximum on $(t_1,t_2)$. Indeed, the fact that, for each $i=1,2$, both $\C$ and $\C_i$ are Graf decompositions at $q(t_i)$, it must hold that 
\[ J_{\C}^E(t_i) + \delta_2^{k_0} = J_{\C_i}^E(t_i) + \delta_2^{|\C_1|} .  \]
Then, from $ J_{\C}^E(t) \leq  J(t)$ we have 
\[ J_{\C}^E(t_i) \leq J(t_i) + \delta_2^{|\C_i|} - \delta_2^{k_0} .  \]
Then, from the choice of $s \geq \max\{t_0,u_0\}$ and $t_p > s$ we have $|J(t_p) - J(t_i)| < \delta_1^n$. Hence $J(t_i) \leq J(t_p) + \delta_1^n$. Moreover, since $\C$ is a Graf partition at $q(t_p)$ we have $J^E_{\C}(t_p) \geq J(t_p) + \delta_1^n - \delta_1^{k_0}$. It follows that 
\begin{equation}\label{eq:JCtiBound}
	J^E_{\C}(t_i) \leq J(t_p) + \delta_1^n + \delta_2^{|\C_i|} - \delta_2^{k_0} \leq J^E_{\C}(t_p) + \delta_1^{k_0} + \delta_2^{k_0+1} - \delta_2^{k_0} \leq J^E_{\C}(t_p) - \tfrac12 \delta_2^{k_0}.
\end{equation}
It follows that $J^E_{\C}(t_p) > J^E_{\C}(t_i)$ for $i=1,2$. Thus, $J^E_{\C}(t)$ obtains a local maximum at some time $t_m \in (t_1,t_2)$. At this time $\dot{J}^E_{\C}(t_m) = 0$ and $\ddot{J}^E_{\C}(t_m) \leq 0$.

Now, on the interval $(t_1, t_2)$ all clusters of $\C$ are isolated. It follows from \cref{prop:clustersAreBoundedForce} that $|\dot{p}_C(t)|$ is bounded for all $C\in \C$ and $t\in(t_1,t_2)$.
Using this, a calculation shows
\begin{align*}
	\dot{J}^E_{\C} &= \sum_{C\in\C} \langle q_C, p_C \rangle,\\
	\ddot{J}^E_{\C} &= 2 K_{\C}^{E} + \sum_{C \in \C} \langle q_C, \dot{p}_C \rangle,
\end{align*}
where $K_{\C}^{E} = \sum_{C\in\C} \frac{\norm{p_C}^2}{m_C}$ is the external kinetic energy of this cluster decomposition.
Since the $q_C$ and the $\dot{p}_C$ are uniformly bounded, there exists a constant $B_0$ such that 
$\left|\sum_{C\in \C}\langle q_C, \dot{p}_C\rangle\right| \leq B_0$
on $(t_1,t_2)$. 
Hence $K^E_{\C}(t_m) \leq \tfrac12 B_0$. 
Furthermore, on $(t_1,t_2)$ there is some constant $B_1$ such that 
$$\dot{K}_{\C}^{E} = \sum_{C\in \C} \left|\left\langle \frac{p_C}{m_C}, \dot{p}_C \right\rangle\right| \leq B_1 \sqrt{K_{\C}^{E}(t)}.$$
It follows that
$\left| \frac{d}{dt}\sqrt{K_{\C}^E(t)} \right| \leq \frac{B_1}{2} $
on the entire interval $(t_1,t_2)$. Therefore,
\[\sqrt{K_{\C}^E(t)} \leq  \sqrt{\frac{B_0}{2}} + \frac{B_1}{2}|t-t_m| \leq B_2 \]
for all $t\in(t_1,t_2)$, where $B_2$ is independent of the particular interval chosen sufficiently close to $T^+$.
Finally, 
\[\left|\dot J_{\C}^E(t)\right| \leq 2\sqrt{J_{\C}^E(t)K_{\C}^E(t)}.\]
Since $J_{\C}^E\leq J$ and $J$ remains bounded near $T^+$, there is a constant $B_3>0$ such that
$\left|\dot J_{\C}^E(t)\right| \leq B_3 $
for all $t\in(t_1,t_2)$.
Choose $s$ sufficiently close to $T^+$ such that
\[ B_3(T^+-s)<\tfrac12\delta_2^k.\]
Since $s<t_1<t_p<T^+$, it holds that 
\[ J_{\C}^E(t_p)-J_{\C}^E(t_1) \leq B_3(t_p-t_1) \leq B_3(T^+-s) < \tfrac12 \delta_2^k. \]
This is a contradiction with \eqref{eq:JCtiBound}, which implied
$J_{\C}^E(t_p)-J_{\C}^E(t_1) \geq \tfrac12\delta_2^k.$ 
The contradiction proves that $G_4=\emptyset$.
\end{proof}

\subsection{Generalisation of Painlev\'e's Theorem}
\label{sec:Painleve}

In this section, we return to the specific $-\alpha$-homogenous potentials and bounded forces given in \cref{eq:Motion}. In this context, there is also a generalization of Painlevé's theorem applied to subsystems. 
	
\begin{theorem}[Painlevé for subsystems]\label{thm:PainleveSubSystems} \quad \\
	Let $(q(t), p(t))$ be a solution to \eqref{eqn:ClusterEquation} with a singularity at time $0<T^+<\infty $. Let $\C$ be the final cluster decomposition (recall \cref{def:FinalClusterDecomposition}). Then, for any $C \in \C$ with $\left|C\right|>1$, 
	\begin{align*}
	\lim\limits_{t \uparrow T^{+}} \min\left\{ \norm{q_i(t) - q_j(t)} \mid i, j \in C, i \neq j \right\} = 0.
	\end{align*}
	That is, the bodies approach the collision set, although, perhaps with $q(t) \to \infty$. 
	In particular, if $J^I_C(t)$ is unbounded, then $V_C(t)$ is unbounded.
\end{theorem}

\begin{remark} \quad
	\begin{enumerate}
	    \item Historically, Painlevé's theorem was proven before von Zeipel's theorem. For the generalisation to a subsystem $C$, one has to go the other way around. The reason is, due to conservation of energy for the whole system without external forcing, a singularity can only happen if either some particles collide (in this case the statement is immediate) and/or if the system gets some increase of kinetic energy (leading to bizarre oscillatory behavior or noncollison singularities). However, for a subsystem or system with external force, the total energy is not conserved. So the increase of the kinetic energy could be by \enquote{stealing} energy from some other subsystem. If the subsystem is a cluster from the final cluster decomposition, we will see in the proof that not enough energy can be stolen to get a diverging subsystem (i.e. a noncollision singularity) and oscillatory behaviour is ruled out by the generalization of von Zeipel's theorem (\cref{thm:vonZeipel}).
        \item For only those clusters with a non-collision singularity does the cluster-internal potential necessarily become unbounded. For collision orbits with general constants $Z_{ij}$ of mixed signs, collision can also be reached without divergence of the potential energy. See \cref{rmk:LimitingCollisionExample} for an example of such an orbit. 
	\end{enumerate}
\end{remark}
	
\begin{proof} 
	From \cref{cor:vonZeipelFinalClusters} any cluster $C$ in the final cluster decomposition with $|C|\geq 2$ is either a total collision of the cluster or $\lim\limits_{t\uparrow T^{+}} J_C^I(t) = \infty$. 
	The statement is clear if $C$ has a total collision of the cluster. So assume the latter. 

	If $J_C^I(t)$ is unbounded, then it must happen that $\limsup_{t \uparrow T^{+}} \dot{J}_C^{I}(t) = \infty$.
	If the limit does not exist, that is, $\liminf_{t\uparrow T^+}\dot{J}^I_C(t)<\infty$ then from \cref{lemma:hKEstimates} we must have
	$\int_0^{T^+} K_C^I(t) \diff t$ bounded. But then also $\int_{0}^{T^{+}} \sqrt{K_C^{I}(t)} \diff t $ is bounded. However,
	\begin{align*}
	\frac{\diff}{\diff t} \sqrt{J_C^{I}} = \frac{\dot{J_C^{I}}}{2\sqrt{J_{C}^{I}}} \lesssim  \sqrt{K_C^{I}}.
	\end{align*}
	So if $\int_{0}^{T^{+}} \sqrt{K_C^{I}(t)} \diff t $ is bounded, then also $\limsup_{t \uparrow T^{+}} \sqrt{J_C^{I}(t)} $ is bounded, which is a contradiction to the unboundedness of $J_C^{I}$. Hence, $\liminf_{t\uparrow T^+} \dot{J}_C^I = \infty$ and the limit of $\dot{J}_C^I$ is guaranteed to exist (and be infinite).

	If the limit of $\dot{J}_C(t)$ exists, then there are record times $\trec$ for $\dot{J}_C^{I}$ with arbitrary large $\dot{J}_C^{I}(\trec)$, that is,
	\begin{align*}
	\dot{J}_C^{I}(\trec) = \sup_{t \in [0, \trec]} \dot{J}_C^{I}(t) \gg 0.
	\end{align*}
	Our goal is to show that for any sufficiently large record time $\trec$, the intrinsic energy $h_C^I$ is bounded above and the internal kinetic energy is bounded below by the internal moment of inertia $\dot{J}^I_C$ in such a way that the energy relation establishes
	\begin{equation}\label{eq:potentialKineticRatio}
		\frac{V_C(\trec)}{K^I_C(\trec)} = \frac{h_C^I(\trec)}{K^I_C(\trec)} -1 \leq -\frac{3}{4} + \frac{\const}{ J_C^I(\trec)}.
	\end{equation}   
	Then, we prove there is a record time $\trec < T^+$ after which $\dot{J}^I_C$ is monotonically increasing, that is, $\ddot{J}^I_C(t)>0$ for all $t> \trec$. After this time, all future times are record times, and \cref{eq:potentialKineticRatio} holds for all $t > \trec$. Consequently, the kinetic energy becomes unbounded while the ratio $V_C/K^I_C$ is forced to have an upper negative bound. This forces $V_C \to -\infty$ and thus, the collision set must be approached, proving the claim of the theorem.

	We begin by establishing \cref{eq:potentialKineticRatio}.
	At sufficiently large record time there is the following bound for $J_C^{I}$:
	\begin{align*}
	J_C^{I}(\trec) &\leq J_C^{I}(0) + \trec \cdot \dot{J}_C^{I}(\trec) \leq \const \dot{J}_C^I(\trec)
	\end{align*}
	Using this bound, and from the fact that $\dot{J}_C^I \leq \const \sqrt{J_C^I}\sqrt{K_C^I}$, we obtain a lower bound for $K_C^I$ at sufficiently large record time:
	\begin{align*}
	K_C^{I}(\trec) & \geq \const \frac{\dot{J}^I_C(\trec)^2}{J^I_C(\trec)} \geq \const \dot{J}^I_C(\trec).
	\end{align*}

	For an upper bound on the intrinsic energy, recall from the proof of \cref{lemma:hKEstimates} that, for arbitrary $D > 0$ and $x>0$, $\sqrt{x} \leq \sqrt{D} + \frac{x}{\sqrt{D}}$, together with the bound on the energy and bound on the integral of the kinetic energy from \cref{lemma:hKEstimates}, we have that 
	\begin{align*}
		|h^I_C(\trec)| &\leq |h^I_C(0)| + \const \int_0^{\trec} \sqrt{K^I_C(s)} \diff s \\
			&\leq |h^I_C(0)| + \const \sqrt{D}\, \trec + \const \int_0^{\trec} \frac{K^I_C(s)}{\sqrt{D}}\diff s \\
			&\leq H_D + \frac{\const}{\sqrt{D}} \dot{J}^I_C(\trec),
	\end{align*}
	with $H_D = |h^I_C(0)| + \const \sqrt{D}\,\trec$. 
	
	Putting the bounds for the internal energy and kinetic energy together yields
	\begin{align*}
		\frac{V_C(\trec)}{K^I_C(\trec)} &= \frac{h_C^I(\trec)}{K^I_C(\trec)} -1 \leq \frac{H_D + \frac{\const}{\sqrt{D}} \dot{J}^I_C(\trec)}{\const \dot{J}^I_C(\trec)} - 1 =  \left(\frac{\const}{\sqrt{D}}-1\right) + \frac{\const}{\dot{J}_C^I(\trec)}.
	\end{align*}
	As $D > 0$ is arbitrary, choose it so that \cref{eq:potentialKineticRatio} holds.

	Finally, we demonstrate that there is a sufficiently large $\trec$ after which all $t>\trec$ are record times. 
	Begin with the estimate for $\ddot{J}_C^{I}$:
	\begin{align*}
	\ddot{J}_C^{I} 
	\geq (2-\alpha) K_C^{I} - \alpha |h_C^{I}| - \const \sqrt{J_C^{I}}.
	\end{align*}
	Plugging in the estimates for $K_C^I$ and $|h_C^I|$ from above (and letting $D$ be another arbitrary constant than that chosen to establish \cref{eq:potentialKineticRatio} above), we obtain
	\begin{align*}
		\ddot{J}^I_C(\trec) &\geq (2-\alpha) K_C^I(\trec) - \alpha |h_C^I(\trec)| - \const \sqrt{J_C^I(\trec)} \\
		&\geq (2-\alpha)\const \dot{J}^I_C(\trec) - \const_{D} - \frac{\const}{\sqrt{D}} \dot{J}^I_C(\trec) -\const \sqrt{\dot{J}^I_C(\trec)}.
	\end{align*}
	Choose $D$ sufficiently large so that 
	\[ \const(2-\alpha) - \frac{\const}{\sqrt{D}} > \tfrac12 \const(2-\alpha).\]
	Then 
	\[ \ddot{J}^I_C(\trec) \geq \tfrac12\const(2-\alpha) \dot{J}^I_C(\trec) - \const_{D} - \const \sqrt{\dot{J}^I_C(\trec)}. \]
	But then we can choose $\trec$ sufficient large so that $\ddot{J}^I_C(\trec) \geq \const \dot{J}^I_C(\trec)\gg 1$. After this time $\trec$, $\dot{J}^I_C(t)$ is monotonically increasing. 
\end{proof}

\subsection{Bounds on the exchange of energy}

In this section, we assume the specific $-\alpha$-homogenous potentials and bounded forces given in \cref{eq:Motion}. Under these assumptions, the internal energy $h^I_C$ of an isolated cluster is bounded if the internal moment of inertia $J^I_C$ as $t\to T^+$ is bounded.

\begin{theorem}
	\label{thm:EnergyBound}
	If $C$ is an isolated cluster and $(q(t),p(t))$ a solution on $(0,T^+)$ such that $\lim_{t\uparrow T^+} J^I_C(t) < \infty$, then the internal energy $h^I_C(t)$ is bounded on $[0,T^+]$.  
\end{theorem}

\begin{remark}
	By \cref{cor:vonZeipelFinalClusters}, the boundedness of $J^I_C(t)$ guarantees that there is either not singularity at $T^+$, or cluster $C$ has a collision singularity. It is not yet known whether a noncollision singularity must also have its internal energy bounded. All estimates known in \cite{quaschner2025improbability} for a system with four bodies prove that the total internal energy increases slower than the kinetic energy, but this does not preclude it becoming unbounded.
\end{remark}

\begin{proof}
	If there are no collision singularities on $[0,T^+]$ then the theorem is immediate. So assume there is a collision singularity at $T^+$. Without loss of generality, we can assume that $C$ is a cluster from the final cluster decomposition and hence (by von Zeipel \cref{thm:vonZeipel}) it has to be a total collision (we will reduce the other case to this below). As $J^I_{C}(t)$ is always positive it must limit to $0$ from above, that is, there exists a sequence $(\tau_n)_{n\in\N}$ with $\tau_n \to T^+$ such that $\dot{J}^I_{C}(\tau_n) < 0$. Then, the inequality 
	\[ \int_0^t K_{C}^I(s) \leq \frac{2}{2-\alpha} \dot{J}^I_{\mcal{C}}(t) + \const, \]
	from \cref{lemma:hKEstimates} together with the the bound on the sequence $\dot{J}^I_{C}(\tau_n) < 0$ shows that $\int_0^{T^+} K_{\mcal{C}}^I(t)\,dt < \infty$. In turn, the inequality 
	\[ \limsup_{t \uparrow T^{+}} |h_C^I(t)| \leq |h_C^I(0)| + \const\int_0^{T^+} \sqrt{K_C^I(s)}\,ds\leq |h_C^I(0)| \const \left(\int_0^{T^{+}} K_C^I(s)\,ds\right)^{1/2} \]
	shows that the internal energy for a collision subcluster $h^I_{\mcal{C}}(t)$ is bounded on $t\in[0,T^+]$. 
	
	If we have an isolated cluster composed out of different colliding subclusters, the internal energy splits into a sum of the internal energies of the colliding subclusters (which is bounded by the above argument), the potential energy between the clusters (which is bounded, as there is a minimal distance between the colliding clusters) and the external kinetic energy (which is also bounded due to the bounded forces between the clusters).
\end{proof}

\section{Asymptotics of clusters with collision}
\label{sec:AsymptoticsClustersCollisions}

Let $(q(t),p(t))$ be a solution to \eqref{eq:Motion} defined on $t\in (0,T^+)$ and let $\C$ be the final cluster decomposition (see \cref{def:FinalClusterDecomposition}). Suppose that $C\in\C$ is a cluster with a collision singularity as $t\to T^+$. In particular, as $C$ is a cluster in the final cluster decomposition, the bodies in $C$ terminate in a total collision with one another at $t = T^+$. We can assume without loss of generality that $C$ is isolated on $(0,T^+)$ so that, as demonstrated in \cref{prop:clustersAreBoundedForce}, the dynamics of the cluster $C$ on $(0, T^+)$ is of the form 
\[ \dot{q}^I_i = \frac{p^I_i}{m_i},\qquad \dot{p}^I_i = - \partial_{q^I_i} V_C(q_i^I) + \tilde{F}^C_i \]
with $\tilde{F}^C_i(t)$ bounded on $(0,T^+)$. 
In this section we perform the McGehee blowup \cite{mcgehee1974triple,moeckel2023total} of the total collision to show that foundational results about the asymptotic behaviour of various system quantities hold in the general motion of \cref{eq:Motion}. 

The McGehee blowup has the effect of replacing the collision singularity in phase space by a manifold, aptly named the \emph{collision manifold}. This is achieved by introducing appropriate polar like coordinates. Namely, introduce the following coordinates for the cluster $C$: 
\begin{equation}
	\begin{aligned}
		r &:= \sqrt{2 J^I(q^I)} \\
		Q &:= r^{-1} q^I \\
		\nu &:= r^{\alpha/2} \left\langle Q, p^I \right\rangle \\
		\omega &:= r^{\alpha/2} \mcal{M}^{-1}_C p^I - \nu Q
	\end{aligned}
\end{equation}
where $\mcal{M}_C$ is the mass matrix for the cluster $C$. Indeed, these are polar like coordinates as the inverse transformation is 
\[ q^I = r Q,\qquad p^I = r^{-\alpha/2} \mcal{M}_C (\omega + \nu Q) \] 
with 
\[\langle \omega, Q\rangle_C := \langle \omega, \mcal{M}_C Q \rangle = 0,\qquad \norm{Q}_C^2 := \langle Q, Q\rangle_C = 1. \]
In particular, $\omega$ is the associated angular velocity and $\nu$ the associated radial velocity.

The $-\alpha$-homogeneity of $V_C(q)$ together with Euler's homogeneous function theorem provides the useful identities
\begin{equation}
    V_C(r Q) = r^{-\alpha} V_C(Q),\quad \nabla V_C(r Q) = r^{-\alpha-1}\nabla V_C(Q)_,\quad \langle Q, \nabla V_C(Q) \rangle = -\alpha V_C(Q). 
\end{equation}
Using these homogeneity identities, the key cluster quantities in these coordinates are computed as
\begin{equation}\label{eqn:clusterQuantitiesMcGehee}
	\begin{aligned}
			J^I_C &= \tfrac12 r^2 \\
			K^I_C &= \tfrac12 r^{-\alpha}(\norm{\omega}_C^2 + \nu^2) \\
			h^I_C &= r^{-\alpha}\left( \tfrac12\left( \norm{\omega}_C^2 + \nu^2 \right) + V_C(Q)  \right).
	\end{aligned}
\end{equation}
The overall result of the transformation to McGehee coordinates is to have replaced the total collision singularity of cluster $C$ (occurring when $J^I(q) = 0$ by \cref{cor:vonZeipelFinalClusters}) with the entire set $r = 0$. 
In fact, the internal energy $h^I$ is bounded on $(0, T^+)$ as shown in \cref{thm:EnergyBound}. Consequently, as $r\to 0$ the collision orbit $(q(t),p(t))$ must approach the manifold 
\[ \operatorname{CM}_C := \left\{(r,Q,\nu,\omega)\,|\, r=0,\, \tfrac12 (\norm{\omega}_C^2 + \nu^2) + V_C(Q) = 0 \right\} \]
henceforth known as the \emph{collision manifold of cluster $C$}.

A computation shows the dynamics in the McGehee coordinates are 
\begin{equation}
	\begin{aligned}
		\dot{r} &= r^{-\tfrac{\alpha}{2}} \nu \\
		\dot{Q} &= r^{-\tfrac{\alpha}{2}-1} \omega \\
		\dot{\nu} &= r^{-\tfrac{\alpha}{2}-1}\left( \tfrac12 \alpha \nu^2  + \norm{\omega}_C^2 + \alpha V_C(Q) \right) + r^{\tfrac{\alpha}{2}}\left\langle Q, \tilde{F} \right\rangle \\
		\dot{\omega} &= -r^{-\tfrac{\alpha}{2}-1} \left( \tfrac12(2-\alpha) \nu\omega + \left(\norm{\omega}_C^2 + \alpha V_C(Q)\right) Q + \mcal{M}_C^{-1} \nabla V_C(Q)  \right) \\
		&\qquad+ r^{\tfrac{\alpha}{2}}\left(\mcal{M}_C^{-1} \tilde{F} - \left\langle Q,\tilde{F} \right\rangle Q \right). 
	\end{aligned}
\end{equation}

After blowup, the final step in the McGehee process is desingularisation. Introduce the new time variable $\tau$ which satisfies $dt = r^{\tfrac{\alpha}{2}+1} d\tau$. Denote a derivative with respect to $\tau$ with a prime ${}^\prime$. The corresponding equations in the new time are 
\begin{equation}
	\begin{aligned}
		r^\prime &= r \nu \\
		Q^\prime &= \omega \\
		\nu^\prime &= \tfrac12 \alpha \nu^2  + \norm{\omega}_C^2 + \alpha V_C(Q) + r^{\alpha+1} \langle Q, \tilde{F}\rangle \\
		\omega^\prime &= -\left( \tfrac12(2-\alpha) \nu\omega + \left(\norm{\omega}_C^2 + \alpha V_C(Q)\right) Q + \mcal{M}_C^{-1} \nabla V(Q)  \right) \\
		&\qquad+ r^{\alpha + 1}\left(\mcal{M}_C^{-1} \tilde{F} - \langle Q,\tilde{F} \rangle Q \right). 
	\end{aligned}
\end{equation}

As a consequence of the desingularisation, we now have a well-defined vector field which is merely a rescaling of the original outside of $r=0$, and well-defined when $r = 0$. In particular, a short calculation shows the collision manifold $\operatorname{CM}_C$ is an invariant manifold of the vector field. Moreover, the rest points on the collision manifold are precisely the set of points
\[ \operatorname{CM}_C^0 = \left\{(0,Q,\nu,\omega)\in \operatorname{CM}_C \mid \omega = 0,\ V_C(Q) = -\tfrac12 \nu^2,\  \mcal{M}_C^{-1} \nabla V_C(Q) + \alpha V_C(Q) Q = 0.\right\} \]
In fact, the last condition on $Q$ is precisely the condition for $Q$ to be a central configuration of the cluster $C$.

The following lemma generalises Proposition~3.1 from \cite{moeckel1983orbits} and establishes $\nu$ as a Lyaponov like quanity on the collision manifold $\operatorname{CM}_C$.
\begin{lemma}
	\label{lem:LjapunovFunction_nuC1}
	For any cluster with a total collision singularity as $t\to T^+$, the function $\nu$ is monotonically increasing on the collision manifold $\operatorname{CM}_C$ and even strictly increasing outside the set of rest points $\operatorname{CM}_C^0$.
\end{lemma}

\begin{proof}
	On the collision manifold the dynamics for $\nu$ are 
	\[ \nu^\prime = \tfrac12 \alpha \nu^2 + \norm{\omega}_C^2 + \alpha V_C(Q) = \tfrac12(2-\alpha) \norm{\omega}_C^2 \geq 0. \]
    Now, suppose there are two times $\tau_1 < \tau_2$ for which $\nu(\tau_1) = \nu(\tau_2) $, that is, $\nu$ has not increased between these values. We will show that this implies the solution is in the set of equilibria on the collision manifold.

    First, integration gives
    \[ \nu(\tau_2) - \nu(\tau_1) = \frac{2-\alpha}{2} \int_{\tau_1}^{\tau_2} \norm{\omega(\tau)}_C^2 \, d\tau. \]
    Equality of $\nu(\tau_2) = \nu(\tau_1)$ can only occur if $\omega(\tau) = 0$ on the domain $(\tau_1,\tau_2)$. Then $\omega'(\tau) = 0$. Moreover, $Q'(\tau) = \omega(\tau) = 0$. Hence, the orbit $(Q(\tau), \nu(\tau), \omega(\tau))$ is constant on $(\tau_1, \tau_2)$. By uniqueness of solutions, the entire orbit is in the set of equilibria.
\end{proof}

We now turn to the asymptotic behaviour of the collision orbit as it approaches the collision manifold. We will prove a proposition stating that any collision orbit will approach the set of rest points on the collision manifold.

\begin{proposition}
	\label{prop:OmegaLimitSet}
	Let $C$ be a cluster and suppose that $(q^I(t),p^I(t))$ is a solution of \eqref{eq:Motion} defined on a maximal interval $t \in [0, T^{+})$ such that $T^{+} < \infty$ 
	and
	\begin{align*}
		\lim\limits_{t \uparrow T^{+}} J_{C}^{I}(t) &= 0, \\
		\limsup_{t \uparrow T^+} V_C(Q^I(t)) &= -a < 0. 
	\end{align*}
	Then the rescaled solution $(\tilde{q}^I(\tau), \tilde{p}^I(\tau))$ is defined on $[0,b)$ and, in McGehee coordinates,
	\begin{enumerate}[(i)]
		\item $b=\infty$, i.e. the solution is defined for all times $\tau$,
		\item for sufficiently large $\tau$, there exists $\nu_2 > \nu_0 > \nu_1 > 0$ and $c_1,c_2>0$ such that,
        \[-\nu_2 < \nu(\tau) < -\nu_1,\quad\text{and}\quad \nu(\tau)\to-\nu_0,\quad\text{and}\quad c_1 e^{-\nu_2 \tau} \leq r(\tau) \leq c_2 e^{-\nu_1 \tau}. \]
        In particular, $r(\tau)$ converges exponentially to 0.
	\end{enumerate}
	Additionally, if for all sufficiently large $\tau$, the normalised configuration $Q(\tau)$ remains in a compact collision-free subset $U \subset S_{C}\setminus \Delta$ with
	\begin{align*}
		S_{C} = \left\{ Q \in \R^{d\cdot \left|C\right|} \mid \scalprod{Q}{\mathcal{M}_C Q} = 1 \right\}
	\end{align*}
	 then
	\begin{enumerate}[(i)]
		\setcounter{enumi}{2}
		\item the omega-limit set is a nonempty compact subset of the set of rest points in the collision manifold. In particular,
        \[ \omega(\tau) \to 0,\quad V_C(Q(\tau)) \to -\frac{1}{2} \nu_0^2 \]
        and $Q(\tau)$ approaches the set of central configurations.
	\end{enumerate}
\end{proposition}

The proof is analogous to the proof in the appendix of \cite{moeckel2023total} adapted to the additional bounded force terms and for a more general class of potentials.

\begin{proof} 
	Using the desingularised differential equation for $\nu$ and the definition of the internal energy of the cluster in \eqref{eqn:clusterQuantitiesMcGehee}, we have the two useful equations,
    \begin{align}
        \nu' &= -(2-\alpha)(\tfrac12 \nu^2 + V_C(Q)) + 2r^\alpha h^I + r^{\alpha+1}\langle Q, \tilde{F} \rangle\label{eq:EstimatenuC1Prime} \\
        \nu' &= \frac{2-\alpha}{2} \norm{\omega}_C^2 + \alpha r^\alpha h^I + r^{\alpha+1} \langle Q, \tilde{F} \rangle.\label{eq:EstimatenuC2Prime}
    \end{align}
	Since $\tilde{F}$ and $h^I$ are bounded ( \cref{prop:clustersAreBoundedForce,thm:EnergyBound}), then there is a constant so that 
    \begin{equation} \label{eq:boundedEnergyAndForceTerm}
        |\alpha r^\alpha h^I + r^{\alpha+1}\langle Q,\tilde{F}\rangle| \leq \const r^\alpha.
    \end{equation}
    Applying this estimate to \eqref{eq:EstimatenuC1Prime} and \eqref{eq:EstimatenuC2Prime} we derive
    \begin{align*}
        \nu' &\geq -\frac{2-\alpha}{2} \nu^2 - (2-\alpha) V_C(Q) - \const r^\alpha, \\
        \nu' &\geq - \const r^\alpha.
    \end{align*}

	By assumption on the limiting behaviour of $V_C$ and $J^I_C$, for any $\epsilon_0$ such that $ 0 < \epsilon_0 < \tfrac12 \sqrt{a(2-\alpha)}$, there exists a $\tau_0 \in (0,b)$ such that both 
	\[-(2-\alpha)V_C(Q(\tau)) > 2\epsilon_0^2,\quad \text{ and, }\quad \const r^\alpha(\tau) < \epsilon_0^2,\]
	for all $\tau \in (\tau_0, b)$.  
	It follows that 
	\[ \nu'(\tau) \geq \epsilon_0^2 - \frac{2-\alpha}{2}\nu^2(\tau), \]
	for all $\tau \in (\tau_0, b)$.

	Now, suppose there is a value $\tau_1 \in (\tau_0, b)$ such that $\nu(\tau_1) \in \left[-\frac{\epsilon_0}{\sqrt{2-\alpha}}, \frac{\epsilon_0}{\sqrt{2-\alpha}} \right]$. Then $\nu'(\tau_1) \geq \tfrac12 \epsilon_0^2$. Furthermore, $\nu'(\tau) \geq \tfrac12 \epsilon_0^2$ will continue to hold for all $\tau \in (\tau_1,\tau_2)$ for some value $\tau_2$ at which $\nu(\tau_2) = \tfrac{1}{\sqrt{2-\alpha}} \epsilon_0$. For all $\tau > \tau_2$, we are forced to have  $\nu(\tau) > \tfrac{1}{\sqrt{2-\alpha}} \epsilon_0$ as the positivity of $\nu'$ when $\nu = \tfrac{1}{\sqrt{2-\alpha}} \epsilon_0$ ensures $\nu(\tau)$ can never decrease back below $\tfrac{1}{\sqrt{2-\alpha}} \epsilon_0$. In summary, if $\nu$ enters the strip $\left[-\frac{\epsilon_0}{\sqrt{2-\alpha}}, \frac{\epsilon_0}{\sqrt{2-\alpha}} \right]$ then necessarily there is a positive lower bound for $\nu$ if it leaves the strip. 
	
	Note that as long as we are inside the strip, $r$ decreases at most as $r(\tau) \geq r(\tau_1) \cdot e^{-\frac{\epsilon_0}{\sqrt{2-\alpha}} (\tau-\tau_1)}$, so the solution would be defined for arbitrary large $\tau$. However, in order for $\nu$ to cross the strip $\left[-\frac{\epsilon_0}{\sqrt{2-\alpha}}, \frac{\epsilon_0}{\sqrt{2-\alpha}} \right]$ with $\nu' \geq \frac{\epsilon_0^2}{2}$ it only takes a finite time. So the strip will eventually be crossed.
	
	However, a positive lower bound on $\nu$ is in contradiction with the assumption that $r(\tau) \to 0$. The contradiction arises because $r' = r \nu$ forces $\nu(\tau) < 0$ for $\tau$ arbitrarily close $b$. Thus, it must hold that $\nu \leq -\frac{\epsilon_0}{\sqrt{2-\alpha}} =: -\nu_1$ for $\tau > \tau_0$. Without loss of generality, we assume $\tau_0 = 0$ by shifting the origin of time if need be.  

	As $\nu(\tau) \leq -\nu_1$ it follows that $r'(\tau) \leq - \nu_1 r(\tau)$, thus 
	\[ r(\tau) \leq r_0 \exp(-\nu_1 \tau). \]
	In other words, $r(\tau)$ decays at least as fast as an exponential.
		
	Now, to establish that $r(\tau)$ is no more than exponentially contracting, it is sufficient to establish a lower bound for $\nu(\tau)$. Using the fact above that $\nu' \geq -\const r^\alpha$, we have,
	\begin{align*}
		\nu'(\tau)  \geq -\const r_0^\alpha \exp(-\nu_1\alpha \tau).
	\end{align*}
	Integrating the inequality yields
	\begin{align*}
		\nu(\tau) &\geq  \nu(0) + \const \frac{r_0^\alpha}{\nu_1 \alpha}\left( \exp(-\nu_1 \alpha \tau) -1  \right) \geq \nu(0)-  \const \frac{r_0^\alpha}{\nu_1 \alpha}=: -\nu_2
	\end{align*}
	This is a lower bound for all times $\tau$. It follows that there is $c_1,c_2 >0$ such that
    \[ c_1 e^{-\nu_2 \tau} \leq r(\tau) \leq c_2 e^{-\nu_1 \tau}. \]
    In particular, the radius converges at most exponentially to 0 and the solution exists for $\tau \in (0,\infty)$ as desired.

    We now show that $\nu(\tau)$ has a definitive limit $\nu_0$ as $\tau \to \infty$. Indeed, from \eqref{eq:EstimatenuC2Prime} and\eqref{eq:boundedEnergyAndForceTerm} we obtain
    \[ \frac{2-\alpha}{2} \int_0^\tau \norm{\omega(s)}_C^2\,ds \leq \nu(\tau) - \nu(0) + \int_0^\tau \const r^\alpha(s)\, ds. \]
    Then, the fact that $\nu(\tau)$ is bounded and $r(\tau)$ converges exponentially to zero gives 
    \[ \int_0^\infty \norm{\omega(s)}_C^2\,ds < \infty. \]
    Finally, this implies 
    \[ \int_0^\infty |\nu'(\tau)|\,d\tau \leq  \frac{2-\alpha}{2} \int_0^\infty \norm{\omega(s)}_C^2\,ds + \int_0^\infty \const r^\alpha(\tau) < \infty.\]
    Thus, $\nu$ must have a finite limit (as it has bounded variation).

	To show the omega-limit set is a nonempty compact subset of the set of rest points on the collision manifold, we first show that the solution is confined to a compact region of phase space that contains part of the collision manifold. Together with the fact that $r(\tau) \to 0$, this shows that the omega-limit set is a compact and non-empty subset of the collision manifold. Then, \cref{lem:LjapunovFunction_nuC1} will force the omega-limit set to be a subset of the rest points on the collision manifold. In particular $\omega(\tau)\to 0$, $V_C(Q(\tau)) \to -\frac{1}{2} \nu_0^2$ and $Q(\tau)$ approaches the set of central configurations.
	
	To show that the orbit is indeed confined to a compact region of phase space containing $\operatorname{CM}_C$, note that $\nu$ stays within a bounded set. By assumption, $Q$ stays within $U$, which is a compact set. As the potential $V_C$ is bounded on this set, from the energy relation we can also deduce that $\omega$ stays within a bounded set and this implies the claim.
\end{proof}

\begin{remark} \label{rmk:LimitingCollisionExample}
	Here we show the necessity of the assumption on the limit of $V_C(q)$ as $t\to T^+$ in \cref{prop:OmegaLimitSet}. 
	Consider the dynamics of 3 bodies of unit mass on the line $\R$ with potential
	\[ V(q_1,q_2,q_3) = \frac{-1}{|q_1 - q_2|} + \frac{-1}{|q_2 - q_3|} + \frac{4}{|q_1 - q_3|} \]
	So we have taken $Z_{12} = -1, Z_{23} = -1, Z_{31} = 4$. 
	Set initial conditions 
	\[ q^0 = (q_1^0,q_2^0,q_3^0) = (-1,0,1),\quad p^0 = - q^0 = (1,0,-1). \]
	Observe that 
	\[
		\nabla V(q^0) = 0.
	\]
	This means that $q^0$ is a central configuration with $\lambda = V(q^0) = 0$. Hence, there is a homothetic solution of the form $q(t) = \rho(t) q^0$. Indeed, with $p^0 = -q^0$ the radial function is $\rho(t) = (1-t)$. There is a total collision at $T^+ = 1$. Remarkably, however, the momentum is bounded (in fact it is constant) and the solution $q(t)$ extends analytically as a curve through the collision even though it is not a classical solution at $t=T^+$. In particular, $r(t)$ does not approach zero exponentially as in \cref{prop:OmegaLimitSet}. More concretely,
    \[ r(\tau) \sim \frac{2}{\tau^2},\qquad \nu(\tau) \to 0^-. \]
\end{remark}

\begin{remark}\label{rmk:potentialLimit}
	Here we show the necessity in \cref{prop:OmegaLimitSet} of the assumption that $Q(\tau)$ remains bounded away from $\Delta \cap S_{C}$ in order to establish the omega-limit set is a non-empty compact subset of the set of rest points in the collision manifold. Consider five bodies in the plane $\R^2$ of unit mass interacting under the potential 
	\[ V(q_1,q_2,q_3,q_4,q_5) = V_H(q_1, q_2, q_3) + V_V(q_3, q_4, q_5) \]
	with 
	\begin{align*}
		V_H(q_1,q_2,q_3) &= \frac{4}{|q_1 - q_2|}+ \frac{-1}{|q_2-q_3|}+\frac{-1}{|q_1-q_3|} \\
		V_V(q_3,q_4,q_5) &= \frac{1}{|q_4-q_5|} +\frac{-\tfrac12}{|q_4-q_3|} + \frac{-\tfrac12}{|q_5-q_3|}.
	\end{align*}
	This system models five bodies of equal mass and charges given by 
	\[ Z_1 = +2,\quad Z_2 = +2,\quad Z_3 = -\tfrac12,\quad Z_4 = +1,\quad Z_5 = +1. \]
	However, in this system, bodies 1,2 do not experience any force from bodies 4,5 directly; they only influence one another through body 3. 
	
	An explicit particular solution for this potential can be found by imposing additional symmetry on the initial conditions. Position body 3 at the origin, bodies 1 and 2 equidistant from the origin on the horizontal line through the origin, and bodies 4 and 5 equidistant on the vertical line through the origin. See Figure~\ref{fig:LimitingCollisionSubProblem}. More concretely, set 
	\[\begin{aligned}
		q_1 &= (-x_1,0),\quad q_2=(x_1,0),\quad q_3=(0,0),\quad q_4 = (0,x_2),\quad q_5 = (0,-x_2)
	\end{aligned} \]
	If the initial momenta are given by 
	\[ p_1 = (-y_1,0),\quad p_2=(y_1,0),\quad p_3 = (0,0),\quad p_4 = (0,y_2),\quad p_5 = (0,-y_2) \]
	then these symmetric configurations form an invariant sub-problem with dynamics given by 
	\[ \begin{aligned}
		\dot{x}_1 &= y_1,\qquad \dot{y}_1 = 0 \\
		\dot{x}_2 &= y_2,\qquad \dot{y}_2 =  -\frac{1}{4}\frac{1}{|x_2|^3}x_2
	\end{aligned} \]
	The dynamics of $x_1$ and $x_2$ are uncoupled and are given by two distinct Hamiltonian systems. The dynamics for $x_2$ is given by the 1D Kepler problem 
	\[ H_V(x_2,y_2) = \tfrac12 y_2^2 - \frac{1/4}{|x_2|}.  \]
	As is well-known in the Kepler problem, if body 4,5 are initially at rest ($y_2(0) = 0$), then $x_2 = 0$ (a triple collision between bodies 3,4, and 5) will occur at some time $T^+ > 0$ that is dependent on the initial condition $x_2(0)$. Moreover, close to this collision 
	\[ x_2(t) \sim \const (T^+ - t)^{2/3}. \] 
	
	The dynamics for $x_1$ is given by the 1D free particle  $H_H(x_1,y_1) = \tfrac12 y_1^2$. If $y_1(0) = -T^+ x_1(0)$ then the solution is simply
	\[ x_1(t) = x_1(0) \frac{T^+-t}{T^+}, \]
	indicating a collision between bodies 1,2, and 3 occuring at $t=T^+$. 
	
	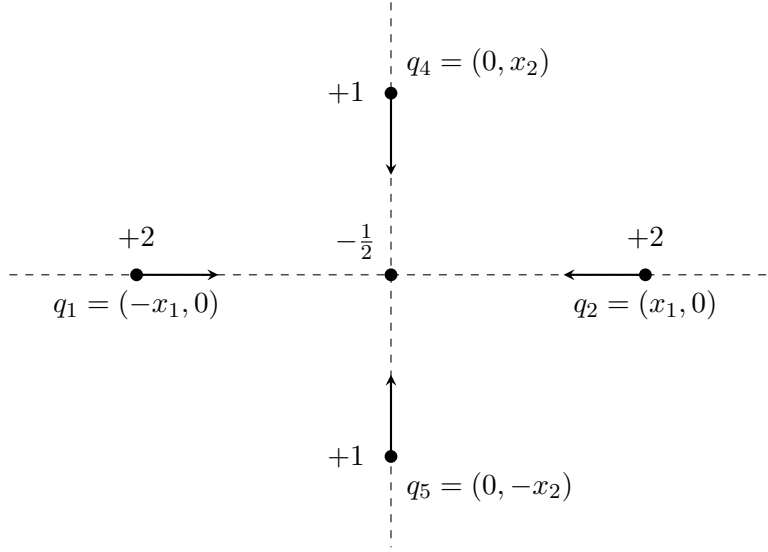
\begin{figure}[ht]
		\centering
		\begin{tikzpicture}[scale=1.2,>=stealth]
			
			\def\xa{2.8}
			\def\xb{2.0}
			\def\arr{0.9}
			
			
			\fill (-\xa,0) circle (2pt) node[below=2pt] {$q_1=(-x_1,0)$};
			\fill ( \xa,0) circle (2pt) node[below=2pt] {$q_2=(x_1,0)$};
			\fill (0,0) circle (2pt) node[below right=2pt] {};
			\fill (0,\xb) circle (2pt) node[above right=2pt] {$q_4=(0,x_2)$};
			\fill (0,-\xb) circle (2pt) node[below right=2pt] {$q_5=(0,-x_2)$};
			
			\draw[->, thick] (-\xa,0) -- ++(\arr,0) node[midway, above] {};
			\draw[->, thick] ( \xa,0) -- ++(-\arr,0) node[midway, above] {};
			\draw[->, thick] (0,\xb) -- ++(0,-\arr) node[midway, right] {};
			\draw[->, thick] (0,-\xb) -- ++(0,\arr) node[midway, right] {};

			\node[above=6pt] at (-\xa,0) {$+2$};
			\node[above=6pt] at ( \xa,0) {$+2$};
			\node[above left=2pt] at (0,0) {$-\tfrac12$};
			\node[left=6pt] at (0,\xb) {$+1$};
			\node[left=6pt] at (0,-\xb) {$+1$};

			\draw[dashed] (-1.5*\xa,0) -- (0,0);
			\draw[dashed] (0,0) -- (1.5*\xa,0);
			\draw[dashed] (0,1.5*\xb) -- (0,0);
			\draw[dashed] (0,0) -- (0,-1.5*\xb);
			
		\end{tikzpicture}
		\caption{The symmetric subproblem described in Remark~\ref{rmk:LimitingCollisionExample}.}
		\label{fig:LimitingCollisionSubProblem}
	\end{figure}

	Hence, with appropriate initial momenta, this particular solution experiences a total collision of all five bodies at $t=T^+$.  As $\tfrac12 r^2 = x_1^2 + x_2^2$, the asymptotics of $x_1, x_2$ imply that $r \sim \sqrt{2} x_2(\tau)$. Moreover, along this solution, $V_H = 0$ and $V_V = -\frac{1}{2x_2}$. Since $V(q) = r^{-1} V(Q)$, we have
    \[ V(Q) = r V(q) = - \frac{r(\tau)}{2 x_2(\tau)} \to -\frac{\sqrt{2}}{2} < 0. \] 
    This shows that the collision orbit satisfies the first hypothesis of \cref{prop:OmegaLimitSet}.
    However, it does not satisfy the second hypothesis. Indeed, the asymptotics of $x_1(t), x_2(t)$ implies 
	\[ \lim_{t \uparrow T^+} Q(t) = \left((0,0),(0,0),(0,0),(0,1/\sqrt{2}),(0,-1/\sqrt{2})\right). \]
	That is, the limiting shape of the configuration involves a triple collision of bodies 1,2,3. 

    It follows that the solution $r(\tau)$ still has exponential contraction, yet, $Q(\tau)$ approaches the collision set $\Delta\cap S_{C}$. Consequently, we cannot conclude that the omega-limit set is a nonempty compact subset of the regular collision manifold, and the Lyapunov argument for convergence to rest points cannot be applied.
\end{remark}

From the above remarks, we see that there are two possible problems for determining the behaviour of the limit to the collision set. The first one is if the potential is ultimately not negative (as in \cref{rmk:LimitingCollisionExample}). The second problem (as in \cref{rmk:potentialLimit}) is when the potential is ultimately negative but the potential has only a finite limit. In this case, the level sets of the potential are not compact on $S_C\setminus\Delta$ and oscillation between different level sets of $V_C$ might be possible.  

However, if all pair potentials are negative (such as when $Z_{ij} = - m_i m_j$ in the typical $n$-body problem) then any approach to collision must force $V_C(Q^I(t)) \to -\infty$. This turns out to prevent oscillation between different level sets of $V_C$, as shown in the following proposition.

\begin{definition}
	A cluster potential $V_C$ is said to be an \emph{attractive pair potential} provided $Z_{ij} < 0$ for all $i\neq j\in C$.
\end{definition}

\begin{proposition}
	Let $C$ be a cluster and suppose that $(q^I(t),p^I(t))$ is a solution of \eqref{eq:Motion} defined on a maximal interval $t \in [0, T^{+})$ such that $T^{+} < \infty$ 
	and $\lim\limits_{t \uparrow T^{+}} J_{C}^{I}(t) = 0$.
	If $V_C$ is an attractive pair potential then:
	\begin{enumerate}[(i)]
		\item $\limsup_{t \uparrow T^+} V_C(Q^I(t)) < 0 $.
		\item $\limsup_{t \uparrow T^+} -V_C(Q^I(t)) < \infty $, i.e. the potential is bounded.
		\item $Q^{I}(t)$ remains in a compact subset of $S_C$ away from collision.
	\end{enumerate}
	Consequently, all statements in \cref{prop:OmegaLimitSet} hold for 
\end{proposition}

\begin{proof}
	(i) is immediate, as $V_C$ is negative on $S_{C} \setminus \Delta$ and the potential diverges when approaching $\Delta$. Statement (iii) will follow from Statement (ii). Indeed, define 
	\begin{align*}
		a:= \sup_{t \in (0, T^{+})} -V_C(Q^{I}(t))< \infty.
	\end{align*}
    Since $V_C(Q) \to -\infty$ as $Q\to \Delta$ there exists a compact set $U$ such that the complement $U^c$ is a neighbourhood of $S_C \cap \Delta$ and such that $V_C(Q) < -a -1$ for all $Q\in U^c\setminus \Delta$. Since $V_C(Q(t)) \geq -a$, then $Q(t)$ must remain in $U$.
	
    What is left to prove is the boundedness of $-V_C(Q^I)$. Let $\nu_1, \nu_2$ be the bounds for $\nu$ proven in \cref{prop:OmegaLimitSet}. 
	
	First of all, note that $\liminf_{t \uparrow T^+} -V_C(Q^I(t)) =\infty$ is impossible, as in this case $\nu'$ would be large and positive after some time, which contradicts the negative upper bound for $\nu$. So the only possibility for $\limsup_{t \uparrow T^+} -V_C(Q^I(t)) = \infty$ is some oscillatory behaviour. \\
	Let $B \gg 0$ be a fixed large value with $B > \liminf_{t \uparrow T^+} -V_C(Q^I(t))$ that might depend on $\nu_1, \nu_2$ and all previously chosen values from \cref{prop:OmegaLimitSet}. 
	First, note that the level sets $V_{C}^{-1}(\{-B\})$ and $V_{C}^{-1}(\{-2B\})$ are two compact disjoint sets, hence they have a positive distance $\eta>0$. If $V_C(Q(\tau))$ were unbounded and oscillating, then the distance $\eta$ must be overcome by $Q(\tau)$ on some interval $[\tau_{-B}, \tau_{-2B}]$ with $V_{C}(Q(\tau_{-B})) = -B$ and $V_{C}(Q(\tau_{-2B})) = -2B$; there would be even infinitely many such intervals. We can estimate the traveled distance in the following way:
	\begin{align*}
		\norm{Q(\tau_{-2B}) - Q(\tau_{-B})} & \leq \int_{\tau_{-B}}^{\tau_{-2B}} \norm{\omega(s)} \diff s.
	\end{align*}
	But from the internal energy relation we can obtain that for times $s$ in this interval
	\begin{align*}
		\norm{\omega(\tau)}_{C}^2 \leq \const B
	\end{align*}
	has to hold, where the constant depends on all terms from above. Due to the norm equivalence on finite-dimensional vector spaces this allows us to deduce that
	\begin{align*}
		\norm{\omega(\tau)} \leq \const \sqrt{B},
	\end{align*}
	where the constant in this equation depends on all of the above quantities and the masses of the particles.
	Thus we get the following estimates:
	\begin{align*}
		\eta \leq \norm{Q(\tau_{-2B}) - Q(\tau_{-B})} & \leq \int_{\tau_{-B}}^{\tau_{-2B}} \norm{\omega(s)} \diff s \leq \left(\tau_{-2B}- \tau_{-B}\right) \const \sqrt{B}
		\intertext{thus}
		\tau_{-2B}- \tau_{-B}                                             & \geq \frac{\eta}{ \const \sqrt{B}}.
	\end{align*}
	As all other quantities in the estimate \eqref{eq:EstimatenuC1Prime} are bounded, we can assume that $B$ is large enough that for all $\tau$ with $V_{C}(Q(\tau)) \leq -B$ we have
	\begin{align*}
		\nu' \geq \frac{B}{2}.
	\end{align*}
	It follows that 
	\begin{align*}
		\nu(\tau_{-2B}) - \nu(\tau_{-B}) & \geq \frac{B}{2} \left(\tau_{-2B}- \tau_{-B}\right) \geq \frac{\eta \sqrt{B}}{ 2\const}.
	\end{align*}
	Consequently, each passage from $-B$ to $-2B$ of the potential would increase $\nu$ by a positive value. However, it has already been established that $\nu$ has bounded variation. Hence, if there were infinitely many such passages, then $\nu$ would not be bounded above, a direct contradiction with the established upper bound for $\nu$. It follows that after some time $V_{C}(Q(\tau))$ has to be bounded below by $-2B$. 
\end{proof}

The following corollary shows how one may combine the geometric approach afforded by the McGehee coordinates to conclude asymptotic results about various cluster quantities.
Recall the asymptotic notation from \cref{def:AsymptoticNotation}.

\begin{corollary}\label{cor:DynamicalEstimates}
	Let $C$ be a cluster and suppose that $(q^I(t),p^I(t))$ is a solution of \eqref{eq:Motion} defined on a maximal interval $t \in [0, T^{+})$ such that $T^{+} < \infty$ 
	and
	\begin{align*}
		\lim\limits_{t \uparrow T^{+}} J_{C}^{I}(t) &= 0, \\
		\limsup_{t\uparrow T^+} V_C(Q) &= -a < 0.
	\end{align*}
    Set $\nu_0$ the limiting value $\nu(\tau) \to -\nu_0 < 0$ guaranteed by \cref{prop:OmegaLimitSet}.
	The following asymptotics hold.
	\begin{enumerate}[(i)]
        \item $ r(t) \sim \left(\frac{2+\alpha}{2} \nu_0\right)^{\frac{2}{2+\alpha}} (T^+ - t)^{\frac{2}{2+\alpha}}.$
		\item $ J^I(t) \sim \frac{1}{2}\left(\frac{2+\alpha}{2} \nu_0\right)^{\frac{4}{2+\alpha}} (T^+ - t)^{\frac{4}{2+\alpha}}$ 
	\end{enumerate}
    Additionally, if for all sufficiently large $\tau$, the normalised configuration $Q(\tau)$ remains in a compact collision-free subset $U \subset S_C\setminus \Delta$ then
    \begin{enumerate}[(i)]
        \setcounter{enumi}{2}
        \item $ h^I(T^+) = \lim_{t\uparrow T^+} h^I(t)$ exists and $|h^I(t) - h^I(T^+)| \lesssim  (T^+-t)^{\frac{2}{2+\alpha}}$.
        \item for all $i,j\in C,\, i\neq j$ it holds that $\norm{q_i(t) - q_j(t)} \asymp (T^+ - t)^{\frac{2}{2+\alpha}}$
        \item $V_C(q^I(t))\sim -\frac{1}{2} \nu_0^2 r(t)^{-\alpha} \sim -\frac{1}{2}\nu_0^2\left(\frac{2+\alpha}{2} \nu_0\right)^{-\frac{2\alpha}{2+\alpha}} (T^+ - t)^{-\frac{2\alpha}{2+\alpha}}.$ 
        \item $K_C^I(p(t)) \sim - V_C(q^I(t))$
        \item $\norm{\frac{p_i(t)}{m_i} - \frac{p_j(t)}{m_j}} \asymp r(t)^{-\alpha/2} \asymp (T^+ - t)^{-\frac{\alpha}{2+\alpha}}$ and also $\norm{\frac{p_i(t)}{m_i} - \frac{p_j(t)}{m_j}} \asymp \norm{q_i(t) - q_j(t)}^{-\alpha/2}.$
        \item $L_C^I(t) = o\left( (T^+ - t)^{\frac{2-\alpha}{2+\alpha}} \right)$. In particular, $L_C^I(t) \to 0$. We can strengthen this statement to $L_C^I(t) = O\left( (T^+ - t)^{\frac{4+\alpha}{2+\alpha}} \right)$.
    \end{enumerate}
	In particular, if $V_C$ is an attractive pair potential then all of the above asymptotics hold.
\end{corollary}
\begin{proof}
	Since $t'(\tau) = r^{\frac{2+\alpha}{2}}$  and $r'(\tau) = \nu r$, then from L'Hopital we have 
    \[ \lim_{\tau \to \infty} \frac{T^+ - t(\tau)}{r(\tau)^{\frac{2+\alpha}{2}}} = \lim_{\tau \to \infty} \frac{-r(\tau)^{\frac{2+\alpha}{2}}}{\frac{2+\alpha}{2}\nu(\tau) r(\tau)^{\frac{2+\alpha}{2}}} = \frac{2}{\nu_0(2+\alpha)}.\]
    Hence, $r(t)^{\frac{2+\alpha}{2}} \sim (\frac{2+\alpha}{2} \nu_0) (T^+-t)$ and (i) follows. Part (ii) follows immediately from the definition $J^I = \tfrac12 r^2$.
    
	Henceforth, assume in addition that $Q(\tau)$ is in a collision-less compact set $U$ of $S_C$. Then from \cref{prop:OmegaLimitSet} we must have $\omega(\tau) \to 0$, $V_C(Q(\tau)) \to -\frac{1}{2} \nu_0^2$. The derivative of the intrinsic energy is given by,  
	\[ {h^I}^\prime(\tau) =  r(\tau) \left\langle \omega(\tau) + \nu(\tau) Q(\tau), \tilde{F}(\tau)  \right\rangle.\]
	As $\omega$, $\nu$, $Q$ and $\tilde{F}$ are bounded, then $\left\langle \omega(\tau) + \nu(\tau) Q(\tau), \tilde{F}(\tau)  \right\rangle$ is uniformly bounded for $\tau$ sufficiently large. Hence, 
    $|{h^I}^\prime(\tau)| \lesssim r(\tau)$ and therefore $h^I_C(T^+) := \lim_{t\to\infty} h^I_C(\tau)$ exists. Using the estimate for $r(\tau)$ proven in (i), the asymptotics for $h^I_C$ claimed in (iii) follows.

	Since $U$ is compact and disjoint from $\Delta$, there exists constants $d^-$ and $d^+$ so that
	\[ 0 < d^- \leq \norm{Q_{i}(\tau) - Q_{j}(\tau)} \leq d^+ < \infty, \]
	uniformly for all pairs with $i\neq j$. Thus 
	\[ \norm{q_i(t)-q_j(t)} = r(t)\norm{Q_i(t) - Q_j(t)} \asymp r(t) \]
	and (iv) follows from (i). 

	By homogeneity $V_C(q^I) = r^{-\alpha} V_C(Q)$. Then the fact that $V_C(Q) \to -\frac{1}{2} \nu_0^2$ and from the asymptotics of $r$ from (i) we have 
	\[ V_C(q^I(t)) \sim -\frac{1}{2}\nu_0^2 r^{-\alpha}(t) \] 
	and so (v) follows from (i).

	As the energy is bounded by (iii), the kinetic energy must balance the potential energy in the limit $t\to T^+$. This establishes (vi).

	We have 
	\[ \frac{p_i}{m_i} - \frac{p_j}{m_j} = r^{-\alpha/2}\left((\omega_i - \omega_j) + \nu(Q_i - Q_j)\right). \]
	From the fact that $\omega \to 0$, $\nu \to -\nu_0$, and $d^+ \geq \norm{Q_i - Q_j} \geq d_i$, it also holds that 
	$ \norm{(\omega_i - \omega_j) + \nu (Q_i - Q_j)} $
	is bounded from above and below. Consequently,
	$ \norm{\frac{p_i}{m_i} - \frac{p_j}{m_j}} \asymp r(t)^{-\alpha/2} $
	and the first part of (vii) follows. The second part of (vii) follows from the estimate of the distances in (iv).

	Finally, 
	\[ L^I_C := \sum_{i\in C} q_i^I\wedge p_i^I  = r^{1-\alpha/2} \sum_{i\in C} m_i Q_i \wedge \omega_i.  \]
	As $\omega_i \to 0$ and $Q_i$ is bounded (being on a sphere), then $L^I_C(t) = o\left( r^{1-\alpha/2}(t)\right)$ and the result follows from (i).

    On the other hand, we can even improve this result. Note that
    \begin{align*}
         \dot{L}^I_C = \sum_{i\in C} q_i^I\wedge \dot{p}_i^I = \sum_{i\in C} q_i^I\wedge \tilde{F}_i^{C}.
    \end{align*}
    So we can estimate
    \begin{align*}
        \norm{L_C^{I}(t)} &= \norm{\int_{t}^{T^+} \dot{L}^I_C(s) ds} \leq \const \int_{t}^{T^+} \sum_{i=1}^{n} \norm{q_i^{I}(s)} ds \leq \const \int_{t}^{T^+} (T^+-s)^{\frac{2}{2+\alpha}} ds \\ & = \const (T^+-s)^{\frac{4+\alpha}{2+\alpha}}.
     \end{align*}
\end{proof}

\begin{remark}
The estimates above recover, in our more general setting, 
that the angular momentum vanishes as collision is approached (at least when all hypothesis are met). This agrees with the early work on the no-infinite-spin problem of Saari and Hulkower \cite{saari1981manifolds} and ElBialy \cite{elbialy1990collision}. However, as pointed out by Moeckel and Montgomery \cite{moeckelNoInfiniteSpin2024}, the vanishing of the angular momentum alone does imply the colliding bodies do not have infinite spin on their way to collision. Indeed, one still needs to show the accumulated geometric rotation is finite.

In \cite{elbialy1990collision}, Elbialy establishes that the angular momentum for the Newtonian $n$-body problem decays as $O\left((T^+-t)^{7/3}\right)$.
The stronger decay in Elbialy's setting stems from a symmetry of the external forces $F_i$ on the cluster (a consequence specifically of the Newtonian problem). However, if one only assumes bounded external forces, the estimate on the angular momentum is essentially sharp.
\end{remark}

\section{Application to improbability of collisions}

As a small application of the above asymptotic theory, we give a simplified proof of the improbability of collision orbits in the $n$-body problem with $-\alpha$-homogeneous potentials for $\alpha \in (0,2)$ and attractive pair potentials. The first one to prove such a result was Saari \cite{saari1971improbability}, but for a smaller range of the exponent $\alpha$. Then, Fleischer and Knauf in \cite{fleischer2019improbability} proved the statement for a larger class of potentials, using the \textbf{Poincaré surface method} from \cite{fleischerImprobabilityWanderingOrbits2019}.However, the sequence of Poincaré surfaces they use is rather complicated. We can give a simpler sequence in our setting and the construction can be used in other situations as well, e.g. for orbits with simultaneous collision and non-collision singularities.

The structure of this section is as follows: we will first recall the Poincar\'e surface method in \cref{thm:PoincareSurfaceMethod}. In order to define the sequence of Poincar\'e surfaces for our purposes, we need two ingredients for the different subclusters at the collision time. One is a set of finite volume in which all particles of the cluster have to be close to the collision time. This will be given in \cref{prop:SetFiniteVolume}. Then, for one particular cluster, we will define a sequence of hypersurfaces in \cref{prop:DefinitionPoincareSurfaceOneSubsystem}, so that for each trajectory all but finitely many will be hit. Finally, in \cref{cor:ImprobabilityCollisions} we will use an exhaustion over the final cluster decomposition and the bounds of some quantities and prove the improbability of all collision orbits.

We use the following version of the Poincaré surface method, which is a special case from the result in \cite{fleischerImprobabilityWanderingOrbits2019}.

\begin{theorem}[Poincar\'e Surface Criterion \cite{fleischerImprobabilityWanderingOrbits2019}] \label{thm:PoincareSurfaceMethod}
	Let $\Omega$ be a volume form on a manifold $M$ and suppose that $X$ is a $C^1$ vector field preserving $\Omega$. Let $(\Sigma_k)_{k\in\N}$ be pairwise disjoint codimension-one submanifolds (possible with boundary) transverse to $X$ such that 
	\begin{enumerate}[(i)]
		\item each $\Sigma_k$ has finitie flux volume; $\mu_k(\Sigma_k) := \int_{\Sigma_k} |\iota_X \Omega| < \infty$, and,
		\item the flux volume shrinks to zero; $\mu_k(\Sigma_k) \to 0$ as $k\to \infty$.
	\end{enumerate}
	Then the set of points $x\in M$ that are wandering in forward time and intersect all but finitely many $\Sigma_k$ has $\Omega$-measure zero. In particular, the set 
    \[\mathcal{T}_{\operatorname{Sing}} := \{x\in M \mid T^+(x) < \infty,\ \mathcal{O}^+(x)\cap \Sigma_k \neq \emptyset \text{ for all sufficiently large $k$}\} \]
    of points $x\in M$ that have finite-time singular orbits and intersect all but finitely many $\Sigma_k$ has $\Omega$-measure zero.
\end{theorem}

We want to apply this theorem to a collision orbit in the $n$-body problem with (possibly) multiple simultaneous collisions. In order to define the sequence of Poincaré surfaces, we need two kinds of sets for the different colliding clusters. For this, we first consider again one final cluster $C$ having a total collision with the dynamics as in \cref{sec:AsymptoticsClustersCollisions}. For convenience, relabel the indices so that $C = \{1,\dots, n\}$ and consider the following coordinates:
\begin{definition}[Jacobi Coordinates on $C$]\label{def:suitableCoords2} \quad \\
    Define the coordinates $(x,y)$ that are connected to the $(q,p)$-coordinates via
    \begin{align*}
		x&=\left(q_1-q_2, q_{\{1,2\}}-q_3, \ldots, q_{\{1,\ldots, n-1\}}-q_n, q_C \right), \\
		y&=\left(\frac{m_2 p_1 - m_1 p_2}{m_1+m_2}, \frac{m_3 p_{\{1,2\}}-m_{\{1,2\}} p_3}{m_{\{1,2\}}+m_3}, \ldots, \frac{m_n p_{\{1,\ldots, n-1\}} - m_{\{1,\ldots, n-1\}} p_n }{m_{\{1,\ldots, n-1\}}+m_{n}}, p_C\right).
	\end{align*}
    In fact, these are canonical coordinates that diagonalise the kinetic energy,
    \begin{align}
        K(p) = \frac{1}{2} \sum_{i=1}^{n-1} \frac{m_{\{1, \ldots, i+1\}}}{m_{\{1, \ldots, i\}} \cdot m_{i+1}} \norm{y_i}^2 + \frac{1}{2} \frac{\norm{y_n}^2}{m_{\{1,\ldots, n\}}}. \label{eq:KJacobiCoordinates}
    \end{align}
\end{definition}

The first step is now to find a set of finite volume in which a colliding subsystem has to stay. This is done in the following proposition.

\begin{proposition} \label{prop:SetFiniteVolume}
	Let $A>0$ be given. Consider a system of $n$ particles $(q,p)$ evolving according to \eqref{eq:Motion}  with attractive pair potentials and bounded external forces satisfying
	\[
	\sup_{t\in(0,T^+),\,j\in C} \norm{F_j(t)} 	\leq A .
	\]
	Then there is a set of finite volume $\mathcal{S}$ (depending on $A$), such that any solution satisfying
	\begin{equation}\begin{split}
			\lim_{t\uparrow T^+} J_C^I(t) = 0,\quad
			\limsup_{t\uparrow T^+}\norm{q_C(t)} \leq A,\quad
			\limsup_{t\uparrow T^+}\norm{p_C(t)} \leq A,\\
			\sup_{t\in(0,T^+)} |h_C^I(t)| \leq A,\quad 
			\limsup_{t\uparrow T^+} -V_C(Q(t))  \leq A
		\end{split}
	\end{equation}
	will stay in this set close to the escape time.
\end{proposition}

\begin{proof}
	Let $\Phi(r) := |\log(r)| r^{-\alpha/2}$ and let $\epsilon >0$ be small. In the coordinates $(x,y)$ defined in \cref{def:suitableCoords2}, consider the set
	\begin{align*}
		\mathcal{S} := \{ (x,y) \mid &\norm{x_n} < A+1,\, \norm{y_n} < A+1, \\ & \forall i \in \{1, \ldots, n-1\}: \norm{x_i} < \epsilon,\, \norm{y_i} < \Phi(\norm{x_i})  \}.
	\end{align*}
	By \cref{cor:DynamicalEstimates} we have for sufficiently late time $t$ close to $T^+$ and for each total collision orbit and for all $i<n$ that $\norm{y_i(t)} \leq C_{\mathrm{orb}} \norm{x_i(t)}^{-\alpha/2}$, where the constant $C_{\mathrm{orb}}$ depends on the particular collision orbit. Since $x_i(t) \to 0$ as $t\to T^+$, then $|\log \norm{x_i(t)}| \to \infty$ and hence $|\log\norm{x_i(t)}| \geq C_{\mathrm{orb}}$ for sufficiently late time $t$. It follows that $\norm{y_i(t)} \leq \Phi(\norm{x_i(t)})$ for every orbit. Thus, close to escape time, every total collision orbit is within the set $\mathcal{S}$.
	
	Define 
	\[ \mathcal{S}_\epsilon = \{(u,v)\in \R^d\times \R^d \mid \norm{u} < \epsilon,\, \norm{v} \leq \Phi(\norm{u}) \}.   \]
	Observe that $\mathcal{S}$ is the Cartesian product of $(n-1)$ copies of $\mathcal{S}_{\epsilon}$ and two balls $B_{A+1} \subset \R^d$ of radius $A+1$. Hence, the volume of $\mathcal{S}$ (with respect to the Lebesgue measure $\lambda^{2n d}$) is 
	\[ \lambda^{2nd}(\mathcal{S}) = [\lambda^d(B_{A+1})]^2 [\lambda^{2d}(\mathcal{S}_{\epsilon})]^{n-1} \]
	Now, compute that
	\begin{equation*}
		\lambda^{2d}(\mathcal{S}_{\epsilon}) = \int_{\norm{u} < \epsilon} \const \Phi(\norm{u})^d \, du = \const \int_{0}^{\epsilon} \rho^{d-1 - d\frac{\alpha}{2}} |\log(\rho)|^{d} \diff \rho  < \infty.
	\end{equation*}
	Here we used Fubini's theorem and a simple change to polar coordinates. The integral is finite, as $d-1 - d\frac{\alpha}{2} = d \frac{2-\alpha}{2} -1 > -1$ and the logarithms can be compensated by a remaining fraction of the power. As $B_{A+1}$ also has finite volume, we can conclude that $\mathcal{S}$ has finite volume.
\end{proof}

Note that the set $\mathcal{S}$ is an open set in the $2\left|C\right|d$-dimensional phase space for the subsystem $C$. We will use the set $\mathcal{S}$ for all but one cluster in the definition of the Poincar\'e surface. For this one cluster, we need to find a sequence of hypersurfaces almost all of which will get hit (close to the escape time). We call this particular cluster again $C$ with the same coordinates and naming conventions as above.

\begin{proposition}\label{prop:DefinitionPoincareSurfaceOneSubsystem}
    Let $A>0$ be given. Consider a system of $n$ particles $(q,p)$ evolving according to \eqref{eq:Motion} with attractive pair potential 
    and bounded external forces satisfying
    \[
        \sup_{t\in(0,T^+),\,j\in C} \norm{F_j(t)} 	\leq A .
    \]
    Let $(x,y)$ be the Jacobi coordinates of \cref{def:suitableCoords2} and let $\Omega$ denote the Liouville volume form. Then, there exists a sequence of pairwise disjoint hypersurfaces $\Sigma_k$, transverse to the time-dependent vector field $X_t$ of \eqref{eq:Motion}, such that every solution satisfying
        \begin{equation}\label{eq:ABounds} \begin{split}
			\lim_{t\uparrow T^+} J_C^I(t) = 0,\quad
		\limsup_{t\uparrow T^+}\norm{q_C(t)} \leq A,\quad
		\limsup_{t\uparrow T^+}\norm{p_C(t)} \leq A,\\
        \sup_{t\in(0,T^+)} |h_C^I(t)| \leq A,\quad 
        \limsup_{t\uparrow T^+} -V_C(Q(t))  \leq A
		\end{split}
	\end{equation}
        intersects $\Sigma_k$ for every sufficiently large $k$.
		Moreover, with $\mu_k(\Sigma_k) := \int_{\Sigma_k}\left|\iota_{X_t}\Omega\right|$, it holds that $\mu_k(\Sigma_k)<\infty$ and there exists a constant $C_A>0$, independent of $k$, such that
        \[
            \mu_k(\Sigma_k)
            \leq
            C_A
            k^{-\left(d(n-1)-1\right)\left(1-\frac{\alpha}{2}\right)}.
        \]
        In particular, $\lim_{k\to\infty}\mu_k(\Sigma_k)=0$.
\\
 Note that here the restriction of $\iota_{X_t}\Omega$ to $\Sigma_k$ is independent of the external forces, and hence of the explicit time-dependence of $X_t$.
\end{proposition}

\begin{proof}
	In Jacobi coordinates, define the hypersurfaces $\Sigma_k$ as
	\begin{align*}
        \Sigma_k := \Bigl\{(x,y) \mid & \norm{x_1}=k^{-1},  \langle x_1,y_1\rangle < 0,\\
        & \norm{x_i}<\const_A k^{-1},\quad i=2,\dots,n-1, \\
        & \norm{x_n}<A+1,\quad \norm{y_n}<A+1, \\ 
        & |h_C^I(x,y)| < A+1,\quad -V_C(x)<\const_A k^\alpha \Bigr\}.
    \end{align*}
	Every orbit satisfying \eqref{eq:ABounds} intersects $\Sigma_k$ for every sufficiently large $k$ as a consequence of \cref{cor:DynamicalEstimates}. Note, to invoke \cref{cor:DynamicalEstimates}, it is important that $x_1 = q_1-q_2$. The distance $\norm{x_1}$ is bounded above and below by $\const \sqrt{J}$, whereas for the other coordinates only the upper bound holds because we compare centers of masses with the position of a particle. Note that the constant $\const_{A}$ depends on the value $A$ because $A$ is an assumed upper bound for $V(Q(t))$, which in turn bounds the minimal relative distance of the particles.

	The vector field $X_t$ is transverse to $\Sigma_k$ for all $k$ due to the condition $\langle x_1, y_1 \rangle < 0$ and that $\dot{x}_1 = \frac{m_1 m_2}{m_1+m_2} y_1$. 

	It remains to establish the claim that $\mu_k(\Sigma_k) < \infty$ and goes to zero as $k$ grows.
	\\
	Since the Jacobi coordinates are canonical coordinates, the Liouville volume form is,
    \[ \Omega = dx_1\wedge\cdots\wedge dx_n \wedge dy_1\wedge\cdots\wedge dy_n.\]
	The volume form $\iota_{X_t} \Omega$ on $\tilde{\Sigma}_k$ is equivalent to the volume form $\langle N, X_t\rangle\,\iota_N \Omega$ where $N$ is the unit normal to $\Sigma_k$. The unit normal is simply $N = \frac{x_1}{\norm{x_1}}$ and hence
    \begin{equation}\label{eq:PoincareDensity}
        \left|\scalprod{N}{X_t}\right| = \left|\scalprod{N}{\dot{x}_1}\right| = \frac{m_1+m_2}{m_1 m_2}
        \left|\scalprod{\frac{x_1}{\norm{x_1}}}{y_1}\right| \leq \frac{m_1+m_2}{m_1 m_2} \norm{y_1}.
    \end{equation}
    In particular, the volume form $\iota_{X_t} \Omega$ is independent of the
    forcing terms $F_j$, thus, time independent. 

	The volume $\mu_k(\Sigma_k)$ can be computed by integrating the density $\left|\scalprod{N}{X_t}\right|$ over the set $\Sigma_k$. We first integrate over the $y$-variables, making use of the condition $|h_C^I|<A+1$. This condition restricts the $y$-variables to a bounded region. Indeed, set
	\[ Y = \left( \frac{y_1}{\sqrt{2M_1}},\dots,\frac{y_{n-1}}{\sqrt{2M_{n-1}}} \right),\qquad M_i := \frac{m_{\{1,\dots,i+1\}}}{m_{\{1,\dots,i\}}\, m_{i+1}}. \] The internal kinetic energy in $Y$ becomes $K_C^I = \sum_{i=1}^{n-1} \norm{Y_i}^2$.
    For fixed $x$, the condition $|h_C^I|<A+1$ therefore restricts $Y$ to
    the spherical shell
    \[ -V_C(x)-(A+1) < \norm{Y}^2 <  -V_C(x)+(A+1), \]
    with the lower bound replaced by $0$ when necessary. Moreover, from \eqref{eq:PoincareDensity} it follows that $\scalprod{N}{X_t} \leq \const \norm{Y}$.
    Consequently, using polar coordinates in $Y\in \R^{d(n-1)}$ and for a fixed $x$,
    \begin{align*}
        \int_{\{Y\in \R^{d(n-1)} \mid |h_C^I(x,Y)|<A+1\}}
        \left|\scalprod{N}{X_t}\right| dY & \leq \const \int_{R_-(x)}^{R_+(x)} \rho^{d(n-1)}\, d\rho.
    \end{align*}
    where
    \[ R_-(x)^2 =  2\max\{0,-V_C(x)-(A+1)\},\qquad  R_+(x)^2 = 2(-V_C(x)+(A+1)). \]
	Introduce $s = \rho^2$. It then follows that
	\begin{align*}
		\const \int_{R_-(x)}^{R_+(x)} \rho^{d(n-1)}\, d\rho &= \const \int_{R_-(x)^2}^{R_+(x)^2} s^{\frac{d(n-1)-1}{2}}\, ds \\
		&\leq \const(R_+(x)^2 - R_-(x)^2) R_+(x)^{d(n-1)-1} \\
		&\leq \const_A k^{\frac{\alpha}{2}(d(n-1)-1)}.
	\end{align*}

    Finally, to complete the volume computation, we must compute the contribution to the volume given by integrating over the $(d-1)$-dimensional sphere
    $\norm{x_1}=k^{-1}$, the variables $x_2,\dots, x_{n-1}$, and the centre of mass coordinates $x_n,y_n$. The sphere $\norm{x_1}=k^{-1}$ is of order $k^{-(d-1)}$, while the volume contribution from each $x_2,\dots, x_{n-1}$ is an order $k^{-d(n-2)}$. The integrations over $x_n$ and $y_n$ contribute only an
    $A$-dependent constant. Hence
    \begin{align*}
        \mu_k(\Sigma_k) &\leq \const_A\, k^{-(d-1)}\times k^{-d(n-2)}\times k^{\frac{\alpha}{2}(d(n-1)-1)} =  \const_A\, k^{-\left(d(n-1)-1\right)
        \left(1-\frac{\alpha}{2}\right)}.
    \end{align*}
    Since $0<\alpha<2$ and $d(n-1)>1$, the exponent is strictly
    negative. Therefore
    \[ \mu_k(\Sigma_k) \to 0,\quad k\to\infty.\]
\end{proof}

As a corollary, we get improbability of collision orbits in this setting. In the proof, we will define Poincar\'e surfaces as Cartesian products of the sets given in \cref{prop:SetFiniteVolume},\cref{prop:DefinitionPoincareSurfaceOneSubsystem}.

\begin{corollary}[Improbability of collision singularities in the attractive $n$-body problem] \label{cor:ImprobabilityCollisions} 
	Consider the $n$-body problem with potentials $V_{i,j}(q) = \frac{Z_{i,j}}{\norm{q_i-q_j}^{\alpha}}$ with $\alpha \in (0,2)$ and $Z_{i,j} \leq 0$ for all $i,j \in \{1, \ldots, n\}$. Then the set
	\begin{align*}
		\operatorname{Coll} := \left\{ (q,p) \mid T^{+}(q,p) < \infty, \lim_{t \uparrow T^{+}(q,p)} J(q(t)) < \infty \right\}
	\end{align*}
	is a set of measure $0$.
\end{corollary}

\begin{proof}
	We use an exhaustion argument and decompose the set $\operatorname{Coll}$ as a countable union according to the (final) cluster decomposition at the collision time and some bound on the forces between the particles from different clusters of the final cluster decomposition, the energies of the colliding subsystems, the center of mass and total momentum of the colliding subsystems and the potential of the limiting configuration on the shape sphere for the solution. For a cluster decomposition $\C$ and a particle $j \in C$, where $C \in \C$ is a cluster, we write
	\begin{align*}
		F_{j, \C} = \sum_{i \not \in C} \alpha Z_{i,j} \frac{q_j-q_i}{\norm{q_i-q_j}^{2+\alpha}}.
	\end{align*}
    Let $\C_{\operatorname{fin}}(q,p)$ be the final cluster decomposition. Then, for $A \in \N$ and $\C \neq \C_{\min} = \{\{1\}, \ldots, \{n\}\}$, define the set
	\begin{align*}
		\operatorname{Coll}_{A, \C} := \Bigl\{ (q,p) & \in  \operatorname{Coll} \mid\,  
        \mathcal{C}_{\operatorname{fin}}(q,p) = \C, \text{ and } \forall C \in \C,\\
        & \sup_{t \in (0, T^{+}), j \in C} \norm{F_{j, \C}(t)} \leq A,\quad \limsup_{t \uparrow T^{+}} \norm{q_C(t)} \leq A, \\
        &\limsup_{t \uparrow T^{+}} \norm{p_C(t)} \leq A,\quad \sup_{t \in (0, T^{+})} \left|h_C^{I}(t)\right| \leq A,\quad \limsup_{t \uparrow T^{+}} -V(Q(t)) \leq A \Bigr\}. 
	\end{align*} 
	Then, from \cref{cor:DynamicalEstimates}, all collision orbits are in one of these sets, that is,
	\begin{align*}
		\operatorname{Coll} = \bigcup_{A \in \N, \C \neq \C_{\min}} \operatorname{Coll}_{A, \C}.
	\end{align*}
	Hence, it is enough to prove that the sets $\operatorname{Coll}_{A, \C}$ are sets of measure zero as the countable union of measure zero sets are measure zero. 
    
	To establish that $\operatorname{Coll}_{A, \C}$ does indeed have measure zero, we use the Poincaré surface method outlined in \cref{thm:PoincareSurfaceMethod}. Let $\C = \{C_1, \ldots, C_k\}$ and without loss of generality, assume $\left|C_1\right|>1$. There must exist a cluster with at least two particles as we cannot have only singleton clusters for a collision singularity otherwise the solution could be extended. As Poincaré surfaces, define a product of sets for each colliding cluster: 
    \begin{enumerate}
        \item For $C_1$, we take the hypersurfaces $\Sigma_{l, C_1}$ defined in \cref{prop:DefinitionPoincareSurfaceOneSubsystem}. Almost all of these hypersurfaces will be hit by the trajectory of the particles in $C_1$, as this \enquote{subsystem} satisfies exactly the equations of motion from the proposition with the force bounds.
        \item For all $i\neq 1$, we use the sets of finite volume $\mathcal{S}_{C_i}$ from \cref{prop:SetFiniteVolume}.
    \end{enumerate}
     Then the sequence of Poincaré surfaces is the sequence 
    \begin{align*}
    	\Sigma_l = \left(\Sigma_{l, C_1} \times \mathcal{S}_{C_2} \times \ldots \times \mathcal{S}_{C_k} \right)_{l \in \N}.
    \end{align*}
	By \cref{prop:SetFiniteVolume,prop:DefinitionPoincareSurfaceOneSubsystem}, any orbit with an initial condition $(q,p) \in \operatorname{Coll}_{A, \C}$ will hit all but finitely many of the hypersurfaces $\Sigma_l$. As the volume of the hypersurface tends to $0$ for $l \to \infty$ and we consider collision orbits that are wandering, it follows from \cref{thm:PoincareSurfaceMethod} that the set $\operatorname{Coll}_{A, \C}$ has measure zero. From this, the corollary follows by exhaustion.
\end{proof}

\section{Conclusion}

This paper extends the classical theorems of Painlev\'e and von Zeipel to subsystems under a general class of $-\alpha$-homogeneous potentials with bounded external forces. Moreover, many classicial asymptotics of collision orbits were extended to this general setting. The key idea is to use the final cluster decomposition to group together bodies which strongly interact with one another in the limit to the singular time $T^+$. This provides a general framework to investigate all singular orbits, especially those containing a subsystem with a non-collision singularity. This approach will be continued in the forthcoming paper \cite{quaschnerGenImprobability2026}, which gives a better understanding of all non-collision singularities where only non-collision subsystems with four particles are considered (and arbitrary collision subsystems). That is, similar to \cref{cor:ImprobabilityCollisions} we will define a subset of singularities, but without the assumption on the boundedness of the moment of inertia and hence allowing also subsystems with non-collision singularities. However, we have to make specific assumptions on the possible final clusters experiencing a non-collision singularity.

Another application of this work could be for proving the improbability of non-collision singularities in the $n$-body problem. This remains an open question from the first problem in Simon's list of  open problems in mathematical physics \cite{simon1984fifteen}. For more than four bodies, the possible dynamics for non-collision orbits are still unknown. However, splitting the system according to its final cluster decomposition seems to be a reasonable approach. In this paper we have shown that the subsystems having a collision singularity behave as one would expect, i.e. similar to the system without perturbing forces. For the approach using Poincaré surfaces, the sets of finite volume for these systems can be useful. 
So the right approach could be to investigate, what orbits for (unperturbed) non-collision singularities are possible and whether these persist (at least approximately) under perturbing forces, e.g., from a colliding subsystem. Note that also in this case, the colliding subsystems are way more well-behaved, as the energy of them can vary, but it is in total a bounded variation and so the other subsystems cannot steel a lot energy. In this sense, the results of this paper clarify the influence of colliding subsystems in non-collision singularities.

The theorems of Painlev\'e and von Zeipel in the context of non-attractive potentials with external forces (respectively for subsystems) are, to the best of our knowledge, new. These theorems should serve as a solid foundation for understanding collision dynamics in charged particle motion, in particular, the possibility of non-collision singularities there. The key examples in \cref{rmk:LimitingCollisionExample,rmk:potentialLimit} hint at more peculiar dynamics then that of attractive potentials.

Moreover, there has been recent work on the no-infinite-spin problem. The problem of no-infinite-spin is to show that the bodies going to collision can only spin finitely many times around their centre of mass before collision. Originally claimed to have been solved in \cite{saari1981manifolds} and \cite{elbialy1990collision}, it was only recently acknowledged that a gap in the proof remained \cite{moeckelNoInfiniteSpin2024}. The no-infinite-spin problem has now been resolved for so-called isolated singularities, first for total collisions in the planar $n$-body problem \cite{moeckelNoInfiniteSpin2024}, then again for partial collisions \cite{gierzkiewiczNoInfiniteSpin2025} and for parabolic orbits \cite{wangProblemInfiniteSpin2025}, for the spatial problem (with additional restriction that the collinear central configuration is not approached) in \cite{pinzariNoInfiniteSpin2026}, and for $\R^d$ (with additional assumption the central configuration approached is of dimension at most $d-1$) in \cite{yuExclusionInfiniteSpin2026}. It remains to show no-infinite-spin for the exceptional cases of non-isolated and collinear in the spatial problem, for partial collisions in $\R^d$, and investigate the possibility of infinite-spin for non-attractive potentials. The asymptotics in this paper are essential for these remaining problems.

\newpage
\bibliographystyle{plain}
\bibliography{Bibliography}

\end{document}